\documentclass[a4paper,12pt]{article}

\usepackage{amsthm,amsmath}
\usepackage[round]{natbib}
\usepackage[colorlinks,citecolor=blue,urlcolor=blue]{hyperref}
\usepackage{amssymb}
\usepackage{graphicx}
\usepackage{amsfonts}
\usepackage{amsbsy}
\usepackage{mathrsfs}
\usepackage{lscape}
\usepackage{float}
\theoremstyle{plain}
\DeclareMathOperator{\sign}{sign}
\usepackage{multirow}
\usepackage{tablefootnote}
\usepackage{threeparttable}
\usepackage{pdflscape}
\usepackage{enumerate}
\usepackage{pdfpages}
\usepackage{authblk}

\numberwithin{equation}{section}
\theoremstyle{plain}
\newtheorem{theorem}{Theorem}
\usepackage[english]{babel}
\newtheorem{lemma}{Lemma}
\newtheorem{remark}{Remark}

\begin{document}
\title{Convex and non-convex regularization methods \\ for spatial point processes intensity estimation}
\author[1]{Achmad Choiruddin}
\author[1, 2]{Jean-Fran\c cois Coeurjolly}
\author[1]{Fr\'ed\'erique Letu\'e}
\affil[1]{Laboratory Jean Kuntzmann, Department of Probability and Statistics, Universit\'e Grenoble Alpes, France}
\affil[2]{Department of Mathematics, Universit\'e du Qu\'ebec \`a Montr\'eal (UQAM), Canada}
\maketitle

\begin{abstract}
This paper deals with feature selection procedures for spatial point processes intensity estimation. We consider regularized versions of estimating equations based on Campbell theorem derived from two classical functions: Poisson likelihood and logistic regression likelihood. We provide general conditions on the spatial point processes and on penalty functions which ensure consistency, sparsity and asymptotic normality. We discuss the numerical implementation and assess finite sample properties in a simulation study. Finally, an application to tropical forestry datasets illustrates the use of the proposed methods.
\end{abstract}

\section{Introduction}
Spatial point pattern data arise in many contexts where interest lies in describing the distribution of an event in space. Some examples include the locations of trees in a forest, gold deposits mapped in a geological survey, stars in a cluster star, animal sightings, locations of some specific cells in retina, or road accidents \citep[see e.g.][]{moller2003statistical, illian2008statistical, baddeley2015spatial}. Interest in methods for analyzing spatial point pattern data is rapidly expanding accross many fields of science, notably in ecology, epidemiology, biology, geosciences, astronomy, and econometrics.

One of the main interests when analyzing spatial point pattern data is to estimate the intensity which characterizes the probability that a point (or an event) occurs in an infinitesimal ball around a given location. In practice, the intensity is often assumed to be a parametric function of some measured covariates \citep[e.g.][]{waagepetersen2007estimating, guan2007thinned, moller2007modern, waagepetersen2008estimating, waagepetersen2009two, guan2010weighted, coeurjolly2014variational}. In this paper, we assume that the intensity function $\rho$ is parameterized by a vector $\boldsymbol \beta$ and has a log-linear specification
\begin{align}
\label{intensity function}
\rho (u;\boldsymbol \beta)=\exp(\boldsymbol \beta^\top \mathbf{z}(u)),
\end{align}
where $\mathbf{z}(u)=\{z_1(u),\ldots,z_p(u)\}^\top$ are the $p$ spatial covariates measured at location $u$ and $\boldsymbol \beta=\{\beta_1,\ldots,\beta_p\}^\top$ is a real $p$-dimensional parameter. When the intensity is a function of many variables, covariates selection becomes inevitable. 

Variable  selection  in  regression  has  a  number  of  purposes:  provide  regularization for good estimation, obtain good prediction, and identify clearly the important variables \citep[e.g.][]{fan2010selective,mazumder2011sparsenet}. Identifying a set of relevant features from a list of many features is in general combinatorially hard and computationally intensive. In this context, convex relaxation techniques such as lasso \citep{tibshirani1996regression} have been effectively used for variable selection and parameter estimation simultaneously. The lasso procedure aims at minimizing:
\begin{align*}
-\log L(\boldsymbol \beta)+\lambda\|\boldsymbol \beta\|_1
\end{align*}
where $L(\boldsymbol \beta)$ is the likelihood function for some model of interest. The $\ell_1$ penalty shrinks coefficients towards zero, and can also set coefficients to be exactly zero. In the context of variable selection, the lasso is often thought of as a convex surrogate for the best-subset selection problem:
\begin{align*}
-\log L(\boldsymbol \beta)+\lambda \| \boldsymbol \beta\|_0.
\end{align*}
The $\ell_0$ penalty $\| \boldsymbol \beta\| _0={\sum_{i=1}^p \mathbb{I}(|\beta_i|>0)}$ penalizes the number of nonzero coefficients in the model.

Since lasso can be suboptimal in model selection for some cases \citep[e.g.][]{fan2001variable, zou2006adaptive, zhang2008sparsity}, many regularization methods then have been developped, motivating to go beyond $\ell_1$ regime to more aggressive non-convex penalties which bridges the gap between $\ell_1$ and $\ell_0$ such as SCAD \citep{fan2001variable} and MC+ \citep{zhang2010nearly}.

More recently, there were several works on implementing variable selection for spatial point processes in order to reduce variance inflation from overfitting and bias from underfitting. \cite{thurman2014variable} focused on using adaptive lasso to select variables for inhomogeneous Poisson point processes. This study then later was extended to the clustered spatial point processes by \cite{thurman2015regularized} who established the asymptotic properties of the estimates in terms of consistency, sparsity, and normality distribution. They also compared their results employing adaptive lasso to SCAD and adaptive elastic net in the simulation study and application, using both regularized weighted and unweighted estimating equations derived from the Poisson likelihood. \cite{yue2015variable} considered modelling spatial point data with Poisson, pairwise interaction point processes, and Neyman-Scott cluster models, incorporated lasso, adaptive lasso, and elastic net regularization methods into generalized linear model framework for fitting these point models. Note that the study by \cite{yue2015variable} also used an estimating equation derived from the Poisson likelihood. However, \cite{yue2015variable} did not provide the theoretical study in detail. Although, in application, many penalty functions have been employed to regularization methods for spatial point processes intensity estimation, the theoretical study is still restricted to some specific penalty functions.

In this paper, we propose regularized versions of estimating equations based on Campbell formula derived from the Poisson and the logistic regression likelihoods to estimate the intensity of the spatial point processes. We consider both convex and non-convex penalty functions. We provide general conditions on the penalty function to ensure an oracle property and a central limit theorem. Thus, we extend the work by \cite{thurman2015regularized} and obtain the theoretical results for more general penalty functions and under less restrictive assumptions on the asymptotic covariance matrix (see Remark~\ref{HL}). The logistic regression method proposed by \cite{baddeley2014logistic} is as easy to implement as the Poisson likelihood method, but is less biased since it does not require deterministic numerical approximation. We prove that the estimates obtained by regularizing the logistic regression likelihood can also satisfy asymptotic properties (see Remark \ref{logistic}). Our procedure is straightforward to implement since we only need to combine the $\texttt{spatstat}$  $\texttt{R}$ package with the two $\texttt{R}$ packages $\texttt{glmnet}$ and $\texttt{ncvreg}$.

The remainder of the paper is organized as follows. Section~\ref{sec2} gives backgrounds on spatial point processes. Section~\ref{sec3} describes standard parameter estimation methods when there is no regularization, while regularization methods are developed in Section~\ref{sec4}. Section~\ref{sec:num} develops numerical details induced by the methods introduced in Sections~\ref{sec3}-\ref{sec4}. Asymptotic properties following the work by \cite{fan2001variable} for generalized linear models are presented in Section~\ref{sec:asy}. Section~\ref{simul} investigates the finite-sample properties of the proposed method in a simulation study, followed by an application to tropical forestry datasets in Section~\ref{sec8}, and finished by conclusion and discussion in Section~\ref{sec9}. Proofs of the main results are postponed to Appendices \ref{sec:auxLemma}-\ref{proof2}.

\section {Spatial point processes}
\label{sec2}

Let $\mathbf{X}$ be a spatial point process on $\mathbb{R}^d$. Let $D \subset \mathbb{R}^d$ be a compact set of Lebesgue measure $|D|$ which will play the role of the observation domain. We view $\mathbf{X}$ as a locally finite random subset of $\mathbb{R}^d$, i.e. the random number of points of $\mathbf{X}$ in $B$, $N(B)$, is almost surely finite whenever $B \subset \mathbb{R}^d$ is a bounded region. Suppose $\mathbf{x}=\{x_1, x_2, \ldots, x_m\}$ denotes a realization of $\mathbf{X}$ observed within a bounded region $D$, where $x_i,i=1, \ldots, m$ represent the locations of the observed points, and $m$ is the number of points. Note that $m$ is random and $0 \leq m < \infty$. If $m=0$ then $\mathbf{x}=\emptyset$ is the empty point pattern in D. For further background material on spatial point processes, see for example \cite{moller2003statistical}.

\subsection{Moments}

The first and second-order properties of a point process are described by intensity measure and second-order factorial moment measure. First-order properties of a point process indicate the spatial distribution events in domain of interest. The intensity measure $\mu$ on $\mathbb{R}^d$  is given by
\begin{align*}
\mu (B)=\mathbb{E} N(B), \mbox{  } B \subseteq \mathbb{R}^d.
\end{align*}
If the intensity measure $\mu$ can be written as
\begin{align*}
\mu (B)=\int_B \rho(u)\mathrm{d}u, \mbox{  } B \subseteq \mathbb{R}^d, 
\end{align*}
where $\rho$ is a nonnegative function, then $\rho$ is called the intensity function. If $\rho$ is constant, then $\mathbf{X}$ is said to be homogeneous or first-order stationary with intensity $\rho$. Otherwise, it is said to be inhomogeneous. We may interpret $\rho(u) \mathrm{d}u$ as the probability of occurence of a point in an infinitesimally small ball with centre $u$ and volume $\mathrm{d}u$.

Second-order properties of a point process indicate the spatial coincidence of events in the domain of interest. The second-order factorial moment measure $\alpha^{(2)}$ on $\mathbb{R}^d \times \mathbb{R}^d$ is given by
\begin{align*}
\alpha^{(2)} (C)=\mathbb{E}{\sum_{u,v \in \mathbf{X}}^{\neq} \mathbb{I}[(u,v) \in C]}, \mbox{  } C \subseteq  \mathbb{R}^d \times \mathbb{R}^d.
\end{align*}
where the $\neq$ over the summation sign means that the sum runs over all pairwise different points $u, v$ in $\mathbf{X}$, and $\mathbb{I}[.]$ is the indicator function. If the second-order factorial moment measure $\alpha^{(2)}$ can be written as
\begin{align*}
\alpha^{(2)} (C)=\int \int \mathbb{I}[(u,v) \in C] \rho^{(2)}(u,v)\mathrm{d}u\mathrm{d}v, \mbox{  } C \subseteq  \mathbb{R}^d \times \mathbb{R}^d,
\end{align*}
where $\rho^{(2)}$ is a nonnegative function, then $\rho^{(2)}$ is called the second-order product density. Intuitively, $\rho^{(2)}(u,v)\mathrm{d}u\mathrm{d}v$ is the probability for observing a pair of points from $\mathbf{X}$ occuring jointly in each of two infinitesimally small balls with centres $u,v$ and volume $\mathrm{d}u, \mathrm{d}v$. Fore more detail description of moment measures of any order, see appendix C in \cite{moller2003statistical}.

Suppose $\mathbf{X}$ has intensity function $\rho$ and second-order product density $\rho^{(2)}$. Campbell theorem \citep[see e.g.][]{moller2003statistical} states that, for any function $k: \mathbb{R}^d \to [0,\infty)$ or $k: \mathbb{R}^d \times \mathbb{R}^d \to [0,\infty)$
\begin{align}
&\mathbb{E} {\sum_{u \in \mathbf{X}} k(u)}={\int k(u) \rho (u)\mathrm{d}u} \label{eq:campbell} \\
&\mathbb{E} {\sum_{u,v \in \mathbf{X}}^{\neq} k(u,v)}=\int{\int k(u,v) \rho^{(2)} (u,v)\mathrm{d}u \mathrm{d}v} \label{eq:campbell2}.
\end{align}

In order to study whether a point process deviates from independence (i.e., Poisson point process), we often consider the pair correlation function given by
\begin{align*}
g(u,v)=\frac{\rho^{(2)}(u,v)}{\rho(u)\rho(v)}
\end{align*}
when both $\rho$ and $\rho^{(2)}$ exist with the convention $0/0=0$. For a Poisson point process (Section \ref{sec:pois}), we have $\rho^{(2)}(u,v)=\rho(u)\rho(v)$ so that $g(u,v)=1$. If, for example, $g(u,v)>1$ (resp. $g(u,v)<1$), this indicates that pair of points are more likely (resp. less likely) to occur at locations $u,v$ than for a Poisson point process with the same intensity function as $\mathbf{X}$. In the same spirit, we can define $\rho^{(k)}$ the $k$-th order intensity function \cite[see][for more details]{moller2003statistical}. If for any $u,v$, $g(u,v)$ depends only on $u-v$, the point process $\mathbf X$ is said to be second-order reweighted stationary.

\subsection{Modelling the intensity function}
\label{sec:intensity}
We discuss spatial point process models specified by deterministic or random intensity function. Particularly, we consider two important model classes, namely Poisson and Cox processes. Poisson point processes serve as a tractable model class for no interaction or complete spatial randomness. Cox processes form major classes for clustering or aggregation. For conciseness, we focus on the two later classes of models. We could also have presented determinantal point processes \citep[e.g.][]{lavancier2015determinantal} which constitute an interesting class of repulsive point patterns with explicit moments. This has not been further investigated for sake of brevity. In this paper, we focus on log-linear models of the intensity function given by (\ref{intensity function}).

\subsubsection{Poisson point process}
\label{sec:pois}
A point process $\mathbf{X}$ on $D$ is a Poisson point process with intensity function $\rho$, assumed to be locally integrable, if the following conditions are satisfied:
\begin{enumerate}
\item for any $B \subseteq D$ with $0 \leq \mu(B)< \infty$, $N(B) \sim Poisson(\mu(B))$,
\item conditionally on $N(B)$, the points in $\mathbf{X} \cap B$ are i.i.d. with joint density proportional to $\rho(u)$, $u \in B$.
\end{enumerate}
A Poisson point process with a log-linear intensity function is also called a modulated Poisson point process \citep[e.g.][]{moller2007modern,waagepetersen2008estimating}. In particular, for Poisson point processes, $\rho^{(2)}(u,v)=\rho(u)\rho(v)$, and $g(u,v)=1, \forall u,v \in D$.

\subsubsection{Cox processes}
\label{sec:cox}
A Cox process is a natural extension of a Poisson point process, obtained by considering the intensity function of the Poisson point process as a realization of a random field. Suppose that $\mathbf{\Lambda}=\{\mathbf{\Lambda}(u):u \in D\}$ is a nonnegative random field. If the conditional distribution of $\mathbf{X}$ given $\mathbf{\Lambda}$ is a Poisson point process on $D$ with intensity function $\mathbf{\Lambda}$, then $\mathbf{X}$ is said to be a Cox process driven by $\mathbf{\Lambda}$ \citep[see e.g.][]{moller2003statistical}. There are several types of Cox processes. Here, we consider two types of Cox processes: a Neyman-Scott point process and a log Gaussian Cox process.\\

\textbf{\textit{Neyman-Scott point processes.}}
Let $\mathbf{C}$ be a stationary Poisson process (mother process) with intensity $\kappa>0$. Given $\mathbf{C}$, let $\mathbf{X}_c, c \in \mathbf{C}$, be independent Poisson processes (offspring processes) with intensity function
\begin{align*}
\rho_c(u;\boldsymbol \beta)=\exp (\boldsymbol \beta^\top \mathbf{z}(u)) k (u-c; \omega)/ \kappa ,
\end{align*}
where $k$ is a probability density function determining the distribution of offspring points around the mother points parameterized by $\omega$. Then $\mathbf{X}= \cup_{c \in \mathbf{C}} \mathbf{X}_c$ is a special case of an {\em inhomogeneous Neyman-Scott point process} with mothers $\mathbf{C}$ and offspring $\mathbf{X}_c, c \in \mathbf{C}$. The point process $\mathbf{X}$ is a Cox process driven by  $\mathbf{\Lambda}(u)=\exp (\boldsymbol \beta^\top \mathbf{z}(u)) {\sum_{c \in \mathbf{C}} k(u-c, \omega})/ \kappa$ \citep[e.g.][]{waagepetersen2007estimating, coeurjolly2014variational} and we can verify that the intensity function of $\mathbf{X}$ is indeed
\begin{align*}
\rho(u;\boldsymbol \beta)=\exp (\boldsymbol \beta^\top \mathbf{z}(u)).
\end{align*}
One example of {\em Neyman-Scott point process} is the {\em Thomas process} where
\begin{align*}
k(u)=(2 \pi \omega^2)^{-d/2} \exp(-\|u\|^2/(2 \omega^2))
\end{align*}
is the density for $N_d(0,\omega^2 \mathbf{I}_d)$. Conditionally on a parent event at location $c$, children events are normally distributed around $c$. Smaller values of $\omega$ correspond to tighter clusters, and smaller values of $\kappa$ correspond to fewer number of parents. The parameter vector $\boldsymbol \psi=(\kappa,\omega)^\top$ is referred to as the interaction parameter as it modulates the spatial interaction (or, dependence) among events.\\

\textbf{\textit{Log Gaussian Cox process.}}
Suppose that $\log \mathbf{\Lambda}$ is a Gaussian random field. Given $\mathbf{\Lambda}$, the point process $\mathbf{X}$ follows Poisson process. Then $\mathbf{X}$ is said to be a log Gaussian Cox process driven by $\mathbf{\Lambda}$ \citep{moller2003statistical}. If the random intensity function can be written as
\begin{align*}
\log \mathbf{\Lambda}(u)=\boldsymbol \beta^\top \mathbf{z}(u)+\boldsymbol \phi(u)- \sigma^2/2,
\end{align*}
where $\boldsymbol \phi$ is a zero-mean stationary Gaussian random field with covariance function $c(u,v;\boldsymbol \psi)=\sigma^2R(v-u;\zeta)$ which depends on parameter $\boldsymbol \psi=(\sigma^2,\zeta)^\top$ \citep{moller2007modern,  coeurjolly2014variational}. The intensity function of this log Gaussian Cox process is indeed given by
\begin{align*}
\rho(u;\boldsymbol \beta)=\exp(\boldsymbol \beta^\top \mathbf{z}(u)).
\end{align*}
One example of correlation function is the exponential form \citep[e.g.][]{waagepetersen2009two}
\begin{align*}
R(v-u;\zeta)=\exp (-\|u-v\| / \zeta), \mbox{ for } \zeta>0.
\end{align*}
Here, $\boldsymbol \psi=(\sigma^2,\zeta)^\top$ constitutes the interaction parameter vector, where $\sigma^2$ is the variance and $\zeta$ is the correlation scale parameter. 

\section{Parametric intensity estimation}
\label{sec3}
One of the standard ways to fit models to data is by maximizing the likelihood of the model for the data. While maximum likelihood method is feasible for parametric Poisson point process models (Section \ref{sec:MLE}), computationally intensive Markov chain Monte Carlo (MCMC) methods are needed otherwise \citep{moller2003statistical}. As MCMC methods are not yet straightforward to implement, estimating equations based on Campbell theorem have been developed \citep[see e.g.][]{waagepetersen2007estimating, moller2007modern, waagepetersen2008estimating, guan2010weighted, baddeley2014logistic}. We review the estimating equations derived from the Poisson likelihood in Section \ref{sec:ee}-\ref{sec:wee} and from the logistic regression likelihood in Section \ref{logilike}.

\subsection{Maximum likelihood estimation}
\label{sec:MLE}
For an inhomogeneous Poisson point process with intensity function $\rho$ parameterized by $\boldsymbol \beta$, the likelihood function is
\begin{align*}
L(\boldsymbol \beta)={\prod_{u \in \mathbf{X} \cap D} \rho(u;\boldsymbol \beta)}\exp\left({\int_D \big(1- \rho(u;\boldsymbol \beta)\big) \mathrm{d}u}\right),
\end{align*}
and the log-likelihood function of $\boldsymbol \beta$ is
\begin{align}
\label{eq:1}
\ell(\boldsymbol \beta)={\sum_{u \in \mathbf{X}  \cap D} \log\rho(u;\boldsymbol \beta)} - {\int_D \rho(u;\boldsymbol \beta)\mathrm{d}u},
\end{align}
where we have omitted the constant term ${\int_D 1\mathrm{d}u}=|D|$. As the intensity function has log-linear form (\ref{intensity function}), (\ref{eq:1}) reduces to
\begin{align*}
\ell(\boldsymbol \beta)={\sum_{u \in \mathbf{X} \cap D} \boldsymbol \beta^\top \mathbf{z}(u)} - {\int_D \exp(\boldsymbol \beta^\top \mathbf{z}(u))\mathrm{d}u}.
\end{align*}
\cite{rathbun1994asymptotic} showed that the maximum likelihood estimator is consistent, asymptotically normal and asymptotically efficient as the sample region goes to $\mathbb{R}^d$.

\subsection{Poisson likelihood}
\label{sec:ee}
Let $\boldsymbol \beta_0$ be the true parameter vector. By applying Campbell theorem (\ref{eq:campbell}) to the score function, i.e. the gradient vector of $\ell(\boldsymbol \beta)$ denoted by $\ell^{(1)}(\boldsymbol \beta)$, we have
\begin{align*}
\mathbb{E} \ell^{(1)}(\boldsymbol \beta)&=\mathbb{E}{\sum_{u \in \mathbf{X} \cap D} \mathbf{z}(u)} - {\int_D \mathbf{z}(u) \exp( \boldsymbol \beta^\top \mathbf{z}(u))\mathrm{d}u}\\
&={\int_D \mathbf{z}(u) \exp( \boldsymbol \beta_0^\top \mathbf{z}(u))\mathrm{d}u} - {\int_D \mathbf{z}(u) \exp( \boldsymbol \beta^\top \mathbf{z}(u))\mathrm{d}u}\\
&={\int_D \mathbf{z}(u) (\exp( \boldsymbol \beta_0^\top \mathbf{z}(u))- \exp( \boldsymbol \beta^\top \mathbf{z}(u)))\mathrm{d}u}=0
\end{align*}
when  $\boldsymbol \beta=\boldsymbol \beta_0$. So, the score function of the Poisson log-likelihood appears to be an unbiased estimating equation, even though $\mathbf{X}$ is not a Poisson point process. The estimator maximizing (\ref{eq:1}) is referred to as the Poisson estimator. The properties of the Poisson estimator have been carefully studied. \cite{schoenberg2005consistent} showed that the Poisson estimator is still consistent for a class of spatio-temporal point process models. The asymptotic normality for a fixed observation domain was obtained by \cite{waagepetersen2007estimating} while \cite{guan2007thinned} established asymptotic normality under an increasing domain assumption and for suitable mixing point processes.

Regarding the parameter $\boldsymbol \psi$ (see Section \ref{sec:cox}), \cite{waagepetersen2009two} studied a two-step procedure to estimate both $\boldsymbol \beta$ and $\boldsymbol \psi$, and they proved that, under certain mixing conditions, the parameter estimates $(\boldsymbol {\hat \beta}, \boldsymbol {\hat \psi})$ enjoy the properties of consistency and asymptotic normality.

\subsection{Weighted Poisson likelihood}
\label{sec:wee}
Although the estimating equation approach derived from the Poisson likelihood is simpler and faster to implement than maximum likelihood estimation, it potentially produces a less efficient estimate than that of maximum likelihood \citep{waagepetersen2007estimating,guan2010weighted} because information about interaction of events is ignored. To regain some lack of efficiency, \cite{guan2010weighted} proposed a weighted Poisson log-likelihood function given by
\begin{equation}
\label{eq:wee}	
	\ell(w;\boldsymbol \beta)={\sum_{u \in \mathbf{X} \cap D} w(u)\log\rho(u;\boldsymbol \beta)} - {\int_D w(u)\rho(u; \boldsymbol\beta)\mathrm{d}u},
\end{equation}
where $w(\cdot)$ is a weight surface. By regarding (\ref{eq:wee}), we see that a larger weight $w(u)$ makes the observations in the infinitesimal region $\mathrm{d}u$ more influent. By Campbell theorem, $\ell^{(1)}(w;\boldsymbol \beta)$ is still an unbiased estimating equation. In addition, \cite{guan2010weighted} proved that, under some conditions, the parameter estimates are consistent and asymptotically normal.

\cite{guan2010weighted} showed that a weight surface $w(\cdot)$ that minimizes the trace of the asymptotic variance-covariance matrix of the estimates maximizing (\ref{eq:wee}) can result in more efficient estimates than Poisson estimator. In particular, the proposed weight surface is
\begin{align*}
w(u)=\{1+\rho(u)f(u)\}^{-1},
\end{align*}
where $f(u)={\int_D \{g(\|v-u\|;\boldsymbol \psi)-1\}\mathrm{d}u}$ and $g(\cdot)$ is the pair correlation function. For a Poisson point process, note that $f(u)=0$ and hence $w(u)=1$, which reduces to maximum likelihood estimation. For general point processes, the weight surface depends on both the intensity function and the pair correlation function, thus incorporates information on both inhomogeneity and dependence of the spatial point processes. When clustering is present so that $g(v-u)>1$, then $f(u)>0$ and hence the weight decreases with $\rho(u)$. The weight surface can be achieved by setting $\hat w(u)=\{1+\hat \rho(u) \hat f(u)\}^{-1}$. To get the estimate $\hat \rho(u)$, $\boldsymbol \beta$ is substituted by $\boldsymbol {\tilde \beta}$ given by Poisson estimates, that is, $\hat \rho(u)=\rho(u;\boldsymbol {\tilde \beta})$. Alternatively, $\hat \rho(u)$ can also be computed nonparametrically by kernel method. Furthermore, \cite{guan2010weighted} suggessted to approximate $f(u)$ by $K(r)- \pi r^2$, where $K(\cdot)$ is the Ripley's $K-$function estimated by
\begin{align*}
\hat K(r)={\sum_{u,v \in \mathbf{X} \cap D}^{\neq} \frac{\mathbb{I}[\|u-v\|\leq r]}{\hat \rho(u)\hat \rho(v)|D \cap D_{u-v}|}}.
\end{align*}

\cite{guan2015quasi} extended the study by \cite{guan2010weighted} and considered more complex estimating equations. Specifically, $w(u) \mathbf{z}(u)$ is replaced by a function $h(u;\boldsymbol \beta)$ in the derivative of (\ref{eq:wee}) with respect to $\boldsymbol \beta$. The procedure results in a slightly more efficient estimate than the one obtained from (\ref{eq:wee}). However, the computational cost is more important and since we combine estimating equations and penalization methods (see Section \ref{sec:pspp}), we have not considered this extension.

\subsection{Logistic regression likelihood}
\label{logilike}
Although the estimating equations discussed in Section \ref{sec:ee} and \ref{sec:wee} are unbiased, these methods do not, in general, produce unbiased estimator in practical implementations. \cite{waagepetersen2008estimating} and \cite{baddeley2014logistic} proposed another estimating function which is indeed close to the score of the Poisson log-likelihood but is able to obtain less biased estimator than Poisson estimates. In addition, their proposed estimating equation is in fact the derivative of the logistic regression likelihood.

Following \cite{baddeley2014logistic}, we define the weighted logistic regression log-likelihood function by
\begin{align}
\label{eq:logilike2}
\ell(w; \boldsymbol \beta) = & {\sum_{u \in \mathbf{X}  \cap D} w(u) \log \left ( \frac {\rho(u;\boldsymbol \beta)} {\delta(u)+\rho(u;\boldsymbol \beta)} \right )} \nonumber \\
& - {\int_D w(u) \delta(u) \log \left (\frac {\rho(u;\boldsymbol \beta) + \delta(u)} {\delta(u)} \right ) \mathrm{d}u},
\end{align}
where $\delta(u)$ is a nonnegative real-valued function. Its role as well as an explanation of the name 'logistic method' will be explained further in Section \ref{sec:logistic}. Note that the score of (\ref{eq:logilike2}) is an unbiased estimating equation. \cite{waagepetersen2008estimating} showed asymptotic normality for Poisson and certain clustered point processes for the estimator obtained from a similar procedure. Furthermore, the methodology and results were studied by \cite{baddeley2014logistic} considering spatial Gibbs point processes.

To determine the optimal weight surface $w(\cdot)$ for logistic method, we follow \cite{guan2010weighted} who minimized the trace of the asymptotic covariance matrix of the estimates. We obtain the weight surface defined by
\begin{align*}
w(u)=\frac{\rho(u) + \delta(u)}{\delta(u)\{1+\rho(u)f(u)\}},
\end{align*}
where $\rho(u)$ and $f(u)$ can be estimated as in Section \ref{sec:wee}.

\section{Regularization techniques}
\label{sec4}
This section discusses convex and non-convex regularization methods for spatial point process intensity estimation.

\subsection{Methodology}
\label{sec:pspp}
Regularization techniques were introduced as alternatives to stepwise selection for variable selection and parameter estimation. In general, a regularization method attempts to maximize the penalized log-likelihood function $\ell(\boldsymbol \theta)-\eta{\sum_{j=1}^p p_{\lambda_j}(|\theta_j|)}$, where $\ell(\boldsymbol \theta)$ is the log-likelihood function of $\boldsymbol \theta$, $\eta$ is the number of observations, and $p_{\lambda}(\theta)$ is a nonnegative penalty function parameterized by a real number $\lambda \geq 0.$

Let $\ell(w;\boldsymbol \beta)$ be either the weighted Poisson log-likelihood function (\ref{eq:wee}) or the weighted logistic regression log-likelihood function (\ref{eq:logilike2}). In a similar way, we define the penalized weighted log-likelihood function given by
\begin{align}
\label{eq:qwee}
	Q(w;\boldsymbol \beta)=\ell(w;\boldsymbol \beta) - |D|{\sum_{j=1}^p p_{\lambda_j}(|\beta_{j}|)},
\end{align}
where $|D|$ is the volume of the observation domain, which plays the same role as the number of observations $\eta$ in our setting, $\lambda_j$ is a nonnegative tuning parameter corresponding to $\beta_j$ for $j=1,\ldots,p$, and $p_\lambda$ is a penalty function described in details in the next section.

\subsection{Penalty functions and regularization methods}
For any $\lambda \geq 0$, we say that $p_\lambda(\cdot): \mathbb{R}^+ \to \mathbb{R}$ is a penalty function if $p_\lambda$ is a nonnegative function with $p_\lambda(0)=0$. Examples of penalty function are the
\begin{itemize}
\item $\ell_2 \mbox{ norm:	} p_\lambda(\theta)=\frac{1}{2} \lambda \theta^2$,
\item $\ell_1 \mbox{ norm:	} p_\lambda(\theta)=\lambda \theta$,
\item $\mbox{Elastic net:	} \mbox{for }0<\gamma<1, p_\lambda(\theta)= \lambda \{\gamma \theta + \frac{1}{2} (1-\gamma)\theta^2\}$,
\item $\mbox{SCAD:	}\mbox{for any }\gamma>2, p_{\lambda}(\theta)=
   \begin{cases}
   \lambda \theta & \mbox{if } \theta \leq \lambda \\
   \frac{\gamma \lambda \theta- \frac {1}{2} (\theta^2 + \lambda^2)}{\gamma-1} & \mbox{if } \lambda \leq \theta \leq \gamma \lambda \\
   \frac{\lambda^2(\gamma^2-1)}{2(\gamma-1)} & \text{if } \theta \geq \gamma \lambda,
   \end{cases}$
\item $\mbox{MC+:	}\mbox{for any }\gamma>1, p_{\lambda}(\theta)=
   \begin{cases}
   \lambda \theta - \frac{\theta^2}{2 \gamma} & \mbox{if } \theta \leq \gamma \lambda\\
   \frac{1}{2}\gamma \lambda^2 & \text{if } \theta \geq \gamma \lambda.\\
   \end{cases}$
\end{itemize}
The first and second derivatives of the above functions are given by Table \ref{table:1}. It is to be noticed that $p'_\lambda$ is not differentiable at $\theta=\lambda, \gamma \lambda$ (resp. $\theta=\gamma \lambda$) for SCAD (resp. for MC+) penalty.

\setlength{\tabcolsep}{3pt}
\renewcommand{\arraystretch}{1.5}
\begin{table}[!ht]
\caption{The first and the second derivatives of several penalty functions.}
\label{table:1}
\begin{center}
\scalebox{0.79}{
\begin{tabular}{ l l l }
\hline
\hline
Penalty & $p'_\lambda(\theta)$ & $p''_\lambda(\theta)$ \\
\hline
\hline
$\ell_2$ & $ \lambda \theta$ &  $\lambda$\\
$\ell_1$ & $\lambda $ & $0$ \\
Elastic net & $\lambda\{(1-\gamma)\theta+\gamma\}$ & $\lambda (1-\gamma)$\\ 
SCAD & 
   $ \begin{cases}
    \lambda & \mbox{if } \theta \leq \lambda \\
   \frac{\gamma \lambda - \theta}{\gamma-1} & \mbox{if } \lambda \leq \theta \leq \gamma \lambda \\
   0 & \text{if } \theta \geq \gamma \lambda \end{cases}$ &
   $ \begin{cases}
   0 & \mbox{if } \theta < \lambda \\
   \frac{-1}{\gamma-1} & \mbox{if } \lambda <\theta < \gamma \lambda\\
   0 & \text{if } \theta > \gamma \lambda
   \end{cases}$ \\
MC+ & 
   $ \begin{cases}
   \lambda - \frac{\theta}{\gamma} & \mbox{if } \theta \leq \gamma \lambda \\
   0 & \text{if } \theta \geq \gamma \lambda
   \end{cases}$ & 
   $ \begin{cases}
   \frac{-1}{\gamma} & \mbox{if } \theta < \gamma \lambda \\
   0 & \text{if } \theta > \gamma \lambda
   \end{cases}$ \\
\hline
\end{tabular}
}
\end{center}
\end{table}

As a first penalization technique to improve ordinary least squares, ridge regression \citep[e.g.][]{hoerl1988ridge} works by minimizing the residual sum of squares subject to a bound on the $\ell_2$ norm of the coefficients. As a continuous shrinkage method, ridge regression achieves its better prediction through a bias-variance trade-off. Ridge can also be extended to fit generalized linear models. However, the ridge cannot reduce model complexity since it always keeps all the predictors in the model. Then, it was introduced a method called lasso \citep{tibshirani1996regression}, where it employs $\ell_1$ penalty to obtain variable selection and parameter estimation simultaneously. Despite lasso enjoys some attractive statistical properties, it has some limitations in some senses \citep{fan2001variable, zou2005regularization, zou2006adaptive, zhang2008sparsity, zhang2010nearly}, making huge possibilities to develop other methods. In the scenario where there are high correlations among predictors, \cite{zou2005regularization} proposed an elastic net technique which is a convex combination between $\ell_1$ and $\ell_2$ penalties. This method is particularly useful when the number of predictors is much larger than the number of observations since it can select or eliminate the strongly correlated predictors together.

The lasso procedure suffers from nonnegligible bias and does not satisfy an oracle property asymptotically \citep{fan2001variable}. \cite{fan2001variable} and  \cite{zhang2010nearly}, among others, introduced non-convex penalties to get around these drawbacks. The idea is to bridge the gap between $\ell_0$ and $\ell_1$, by trying to keep unbiased the estimates of nonzero coefficients and by shrinking the less important variables to be exactly zero. The rationale behind the non-convex penalties such as SCAD and MC+ can also be understood by considering its first derivative (see Table \ref{table:1}). They start by applying the similar rate of penalization as the lasso, and then continuously relax that penalization until the rate of penalization drops to zero. However, employing non-convex penalties in regression analysis, the main challenge is often in the minimization of the possible non-convex objective function when the non-convexity of the penalty is no longer dominated by the convexity of the likelihood function. This issue has been carefully studied. \cite{fan2001variable} proposed the local quadratic approximation (LQA). \cite{zou2008one} proposed a local linear approximation (LLA) which yields an objective function that can be optimized using least angle regression (LARS) algorithm \citep{efron2004least}. Finally, \cite{breheny2011coordinate} and \cite{mazumder2011sparsenet} investigated the application of coordinate descent algorithm to non-convex penalties.

\setlength{\tabcolsep}{3pt}
\renewcommand{\arraystretch}{1.5}
\begin{table}[htb]
\caption{Details of some regularization methods.}
\label{table:rm}
\centering
\begin{threeparttable}
\scalebox{0.79}{
\begin{tabular}{ll}
\hline
\hline
Method & ${\sum_{j=1}^p p_{\lambda_j}(|\beta_j|)}$ \\
\hline
\hline
Ridge &  ${\sum_{j=1}^p \frac{1}{2} \lambda \beta_j^2}$ \\
Lasso & ${\sum_{j=1}^p \lambda |\beta_j|} $ \\
Enet\tnote{*} & $ {\sum_{j=1}^p \lambda \{\gamma |\beta_j|+\frac{1}{2} (1-\gamma)\beta_j^2\}}$\\ 
AL\tnote{*} & ${\sum_{j=1}^p \lambda_j |\beta_j|} $ \\
Aenet\tnote{*} & ${\sum_{j=1}^p \lambda_j \{\gamma |\beta_j|+\frac{1}{2} (1-\gamma)\beta_j^2\}}$\\ 
SCAD & ${\sum_{j=1}^p p_{\lambda}(|\beta_j|)}, \mbox{ with } p_{\lambda}(\theta)=
\begin{cases}
\lambda \theta & \mbox{if } (\theta \leq \lambda)\\
\frac{\gamma \lambda \theta - \frac {1}{2} (\theta^2 + \lambda^2)}{\gamma-1} & \mbox{if } (\lambda \leq \theta \leq \gamma \lambda)\\
\frac{\lambda^2(\gamma^2-1)}{2(\gamma-1)} & \mbox{if } (\theta \geq \gamma \lambda)\\
\end{cases}$\\
MC+ &  $ {\sum_{j=1}^p \Big \{\Big( \lambda |\beta_j| - \frac{\beta_j^2}{2 \gamma}\Big) \mathbb{I}(|\beta_j| \leq  \gamma \lambda)+\frac{1}{2}\gamma \lambda^2 \mathbb{I}(|\beta_j| \geq \gamma \lambda)\Big \} }$\\
\hline
\end{tabular}
}
\begin{tablenotes}
\item[*] Enet, AL and Aenet, respectively, stand for elastic net, adaptive lasso and adaptive elastic net
\end{tablenotes}
\end{threeparttable}
\end{table}

In (\ref{eq:qwee}), it is worth emphasizing that we allow each direction to have a different regularization parameter. By doing this, the $\ell_1$ and elastic net penalty functions are extended to the adaptive lasso \citep[e.g.][]{zou2006adaptive} and adaptive elastic net \citep[e.g.][]{zou2009adaptive}. Table \ref{table:rm} details the regularization methods considered in this study.

\section{Numerical methods}
\label{sec:num}
We present numerical aspects in this section. For nonregularized estimation, there are two approaches that we consider. Weighted Poisson regression is explained in Section \ref{sec:ue}, while logistic regression is reviewed in Section \ref{sec:logistic}. Penalized estimation procedure is done by employing coordinate descent algorithm (Section \ref{sec:cda}). We separate the use of the convex and non-convex penalties in Section \ref{sec:conv} and \ref{sec:nconv}.

\subsection{Weighted Poisson regression}
\label{sec:ue}
\cite{berman1992approximating} developed a numerical quadrature method to approximate maximum likelihood estimation for an inhomogeneous Poisson point process. They approximated the likelihood by a finite sum that had the same analytical form as the weighted likelihood of generalized linear model with Poisson response. This method was then extended to Gibbs point processes by \cite{baddeley2000practical}. Suppose we approximate the integral term in (\ref{eq:1}) by Riemann sum approximation
\begin{align*}
{\int_D \rho(u;\boldsymbol \beta) \mathrm{d}u} \approx {\sum_i^{M} v_i \rho (u_i;\boldsymbol \beta)}
\end{align*}
where $u_i, i=1,\ldots,M$ are points in $D$ consisting of the $m$ data points and $M-m$ dummy points. The quadrature weights $v_i>0$ are such that ${\sum_i v_i}=|D|$. To implement this method, the domain is firstly partitioned into $M$ rectangular pixels of equal area, denoted by $a$. Then one dummy point is placed in the center of the pixel. Let $\Delta_i$ be an indicator whether the point is an event of point process ($\Delta_i=1$) or a dummy point ($\Delta_i=0$). Without loss of generality, let $u_i,\ldots,u_m$ be the observed events and $u_{m+1},\ldots,u_M$ be the dummy points. Thus, the Poisson log-likelihood function (\ref{eq:1}) can be approximated and rewritten as

\begin{align}
\label{eq:qua}
\ell(\boldsymbol \beta) \approx {\sum_i^{M} v_i \{y_i \log\rho(u_i;\boldsymbol \beta) - \rho (u_i;\boldsymbol \beta)\}},
\mbox{ where } y_i=v_i^{-1} \Delta_i.
\end{align}
Equation (\ref{eq:qua}) corresponds to a quasi Poisson log-likelihood function. Maximizing (\ref{eq:qua}) is equivalent to fitting a weighted Poisson generalized linear model, which can be performed using standard statistical software. Similarly, we can approximate the weighted Poisson log-likelihood function (\ref{eq:wee}) using numerical quadrature method by
\begin{align}
\label{eq:wqua}
\ell(w; \boldsymbol \beta) \approx {\sum_i^{M} w_i v_i \{y_i \log\rho(u_i;\boldsymbol \beta) - \rho (u_i;\boldsymbol \beta)\}}.
\end{align}
where $w_i$ is the value of the weight surface at point $i$. The estimate $\hat w_i$ is obtained as suggested by \cite{guan2010weighted}. The similarity beetween $(\ref{eq:qua})$ and $(\ref{eq:wqua})$ allows us to compute the estimates using software for generalized linear model as well. This fact is in particular exploited in the $\texttt{ppm}$ function in the $\mathtt{spatstat}$ $\texttt{R}$ package \citep{baddeley2004spatstat,baddeley2015spatial} with option $\texttt{method="mpl"}$. To make the presentation becomes more general, the number of dummy points is denoted by $\texttt{nd}^2$ for the next sections.

\subsection{Logistic regression}
\label{sec:logistic}
To perform well, the Berman-Turner approximation often requires a quite large number of dummy points. Hence, fitting such generalized linear models can be computationally intensive, especially when dealing with a quite large number of points. When the unbiased estimating equations are approximated using deterministic numerical approximation as in Section \ref{sec:ue}, it does not always produce unbiased estimator. To achieve unbiased estimator, we estimate (\ref{eq:logilike2}) by
\begin{align}
\label{eq:logilike}
\ell(w; \boldsymbol \beta) \approx & {\sum_{u \in \mathbf{X}  \cap D} w(u) \log \left ( \frac {\rho(u;\boldsymbol \beta)} {\delta(u)+\rho(u;\boldsymbol \beta)} \right )} + {\sum_{u \in \mathcal{D}  \cap D} w(u) \log \left (\frac {\delta(u)} {\rho(u;\boldsymbol \beta) + \delta(u)} \right )},
\end{align}
where $\mathcal{D}$ is dummy point process independent of $\mathbf{X}$ and with intensity function $\delta$. The form (\ref{eq:logilike}) is related to the estimating equation defined by \citet[][eq. 7]{baddeley2014logistic}. Besides that, we consider this form since if we apply Campbell theorem to the last term of (\ref{eq:logilike}), we obtain
\begin{align*}
\mathbb{E} {\sum_{u \in \mathcal{D}  \cap D} w(u) \log \left (\frac {\delta(u)} {\rho(u;\boldsymbol \beta) + \delta(u)} \right )}= {\int_D w(u) \delta(u) \log \left (\frac {\rho(u;\boldsymbol \beta) + \delta(u)} {\delta(u)} \right ) \mathrm{d}u},
\end{align*}
which is exactly what we have in the last term of (\ref{eq:logilike2}). In addition, conditional on $\mathbf{X} \cup \mathcal{D}$, (\ref{eq:logilike}) is the weighted likelihood function for Bernoulli trials, $y(u)=1 \{u \in \mathbf{X} \}$ for $u \in \mathbf{X} \cup \mathcal{D}$, with
\begin{align*}
\mathrm{P}\{ y(u)=1 \}=\frac {\rho(u;\boldsymbol \beta)} {\delta(u)+\rho(u;\boldsymbol \beta)}= \frac {\exp \big(- \log \delta(u) + \boldsymbol \beta^\top \mathbf{z}(u) \big)} {1 + \exp \big ( - \log \delta(u) +\boldsymbol \beta^\top \mathbf{z}(u) \big)}.
\end{align*}
Precisely, (\ref{eq:logilike}) is a weighted logistic regression with offset term $- \log \delta$. Thus, parameter estimates can be straightforwardly obtained using standard software for generalized linear models. This approach is in fact provided  in the $\mathtt{spatstat}$ package in $\texttt{R}$ by calling the $\texttt{ppm}$ function with option $\texttt{method="logi"}$ \citep{baddeley2014logistic,baddeley2015spatial}.

In $\texttt{spatstat}$, the dummy point process $\mathcal{D}$ generates $\texttt{nd}^2$ points in average in $D$ from a Poisson, binomial, or stratified binomial point process. \cite{baddeley2014logistic} suggested to choose $\delta(u)=4m/|D|$, where $m$ is the number of points (so, $\texttt{nd}^2=4m$). Furthermore, to determine $\delta$, this option can be considered as a starting point for a data-driven approach \citep[see][for further details]{baddeley2014logistic}.

\subsection{Coordinate descent algorithm}
\label{sec:cda}
LARS algorithm \citep{efron2004least} is a remarkably efficient method for computing an entire path of lasso solutions. For linear models, the computational cost is of order $O(Mp^2)$, which is the same order as a least squares fit. Coordinate descent algorithm \citep{friedman2007pathwise, friedman2010regularization} appears to be a more competitive algorithm for computing the regularization paths by costs $O(Mp)$ operations. Therefore we adopt cyclical coordinate descent methods, which can work really fast on large datasets and can take advantage of sparsity. Coordinate descent algorithms optimize a target function with respect to a single parameter at a time, iteratively cycling through all parameters until convergence criterion is reached. We detail this for some convex and non-convex penalty functions in the next two sections. Here, we only present the coordinate descent algorithm for fitting generalized weighted Poisson regression. A similar approach is used to fit penalized weighted logistic regression.

\subsubsection{Convex penalty functions}
\label{sec:conv}
Since $\ell(w;\boldsymbol \beta)$ given by $(\ref{eq:wqua})$ is a concave function of the parameters, the Newton-Raphson algorithm used to maximize the penalized log-likelihood function can be done using the iteratively reweighted least squares (IRLS) method. If the current estimate of the parameters is $\boldsymbol {\tilde \beta}$, we construct a quadratic approximation of the weighted Poisson log-likelihood function using Taylor's expansion:
\begin{align}
\ell(w;\boldsymbol {\beta}) \approx \ell_Q(w;\boldsymbol {\beta})=- \frac{1}{2M} {\sum_i^M \nu_i (y_i^*-\mathbf{z}_i^\top \boldsymbol{\beta})^2+C(\boldsymbol{\tilde \beta})},
\end{align}
where $C(\boldsymbol{\tilde \beta})$ is a constant, $y_i^*$ are the working response values and $ \nu_i$ are the weights,
\begin{align*}
 \nu_i&=w_i v_i \exp(\mathbf{z}_i^\top \boldsymbol{\tilde \beta})\\
y_i^*&=\mathbf{z}_i^\top \boldsymbol{\tilde \beta}+\frac{y_i - \exp(\mathbf{z}_i^\top \boldsymbol {\tilde \beta})}{\exp(\mathbf{z}_i^\top \boldsymbol {\tilde \beta})}.
\end{align*}

Regularized Poisson linear model works by firstly identifying a decreasing sequence of $\lambda \in [\lambda_{\min},\lambda_{\max}]$, for which starting with minimum value of $\lambda_{\max}$ such that the entire vector $\boldsymbol {\hat \beta} =0$. For each value of $\lambda$, an outer loop is created to compute $\ell_Q(w;\boldsymbol \beta)$ at $\boldsymbol {\tilde \beta}$. Secondly, a coordinate descent method is applied to solve a penalized weighted least squares problem
\begin{align}
\label{eq:glmnet}	
{\displaystyle \min_{\boldsymbol \beta \in \mathbb{R}^{p}} \Omega(\boldsymbol \beta)}={\displaystyle \min_{\boldsymbol \beta \in \mathbb{R}^{p}} \{-\ell_Q(w;\boldsymbol \beta)+\sum_{j=1}^p p_{\lambda_{j}}(|\beta_{j}|)}\}.
\end{align}

The coordinate descent method is explained as follows. Suppose we have the estimate $\tilde \beta_l$ for $l \neq j$, $l,j=1,\ldots,p$. The method consists in partially optimizing ($\ref{eq:glmnet}$) with respect to $\beta_j$, that is
\begin{align*}	
{\displaystyle \min_{\beta_j} \Omega (\tilde \beta_1, \ldots, \tilde \beta_{j-1}, \beta_j, \tilde \beta_{j+1}, \ldots, \tilde \beta_{p})}.
\end{align*}

\cite{friedman2007pathwise} have provided the form of the coordinate-wise update for penalized regression using several penalties such as nonnegative garrote \citep{breiman1995better}, lasso, elastic net, fused lasso \citep{tibshirani2005sparsity}, group lasso \citep{yuan2006model}, Berhu penalty \citep{owen2007robust}, and LAD-lasso \citep{wang2007robust}. For instance, the coordinate-wise update for the elastic net, which embraces the ridge and lasso regularization by setting respectively $\gamma$ to 0 or 1,  is
\begin{align}
\label{eq:upd}
{\tilde \beta_j} \leftarrow \frac {S\left({\displaystyle{\sum_{i=1}^M \nu_jz_{ij}(y_i-{\tilde y_i}^{(j)})}},\lambda \gamma\right)}{{\displaystyle{\sum_{i=1}^M \nu_jz_{ij}^2}+\lambda (1-\gamma)}},
\end{align}
where ${\tilde y_i}^{(j)}={\tilde \beta_0}+ {\sum_{l \neq j}z_{il}{\tilde \beta_l}}$ is the fitted value excluding the contribution from covariate $z_{ij}$, and $S(z,\lambda)$ is the soft-thresholding operator with value

\begin{align}
\label{eq:sz}
S(z,\lambda)=\sign(z)(|z|-\lambda)_+ = 
  \begin{cases}
   z-\lambda & \mbox{if } z>0 \mbox{  and  } \lambda<|z| \\
   z+\lambda & \mbox{if } z<0 \mbox{  and  } \lambda<|z| \\
   0 & \text{if } \lambda \geq |z|.
   \end{cases}
\end{align}

The update (\ref{eq:upd}) is repeated for $j=1,\ldots,p$ until convergence. Coordinate descent algorithm for several convex penalties is implemented in the $\texttt{R}$ package $\texttt{glmnet}$ \citep{friedman2010regularization}. For (\ref{eq:upd}), we can set $\gamma=0$ to implement ridge and $\gamma=1$ to lasso, while we set $0<\gamma<1$ to apply elastic net regularization. For adaptive lasso, we follow \cite{zou2006adaptive}, take $\gamma=1$ and replace $\lambda$ by $\lambda_j=\lambda / |{\tilde \beta_j}|^{\tau}$, where $\boldsymbol {\tilde \beta}$ is an initial estimate, say $\boldsymbol {\tilde \beta}(ols)$ or $\boldsymbol {\tilde \beta}(ridge)$, and $\tau$ is a positive tuning parameter. To avoid the computational evaluation for choosing $\tau$, we follow \citet[][Section 3.4]{zou2006adaptive} and \cite{wasserman2009high} who also considered $\tau=1$, so we choose $\lambda_j=\lambda / |{\tilde \beta_j}(ridge)|$, where $\boldsymbol {\tilde \beta}(ridge)$ is the estimates obtained from ridge regression. Implementing adaptive elastic net follows along similar lines.

\subsubsection{Non-convex penalty functions}
\label{sec:nconv}
\cite{breheny2011coordinate} have investigated the application of coordinate descent algorithm to fit  penalized generalized linear model using SCAD and MC+, for which the penalty is non-convex. \cite{mazumder2011sparsenet} also studied the coordinate-wise optimization algorithm in linear models considering more general non-convex penalties.

\cite{mazumder2011sparsenet} concluded that, for a known current estimate $\tilde{\theta}$, the univariate penalized least squares function $Q_u(\theta)=\frac{1}{2} (\theta-\tilde{\theta})^2+p_{\lambda}(|\theta|)$ should be convex to ensure that the coordinate-wise procedure converges to a stationary point. \cite{mazumder2011sparsenet} found that this turns out to be the case for SCAD and MC+ penalty, but it cannot be satisfied by bridge (or power) penalty and some cases of log-penalty.

\cite{breheny2011coordinate} derived the solution of coordinate descent algorithm for SCAD and MC+ in generalized linear models cases, and it is implemented in the $\mathtt{ncvreg}$ package of $\mathtt{R}$. Let $\tilde{\boldsymbol \beta_l}$ be a vector containing estimates $\tilde \beta_l$ for $l \neq j$, $l,j=1,\ldots,p$, and we wish to partially optimize ($\ref{eq:glmnet}$) with respect to $\beta_j$. If we define $\tilde g_j={\sum_{i=1}^M \nu_jz_{ij}(y_i-{\tilde y_i}^{(j)})}$ and $\tilde \eta_j={\sum_{i=1}^M \nu_jz_{ij}^2}$, the coordinate-wise update for SCAD is 
\[
 {\tilde \beta_j} \leftarrow 
  \begin{cases} 
  \frac {S(\tilde g_j,\lambda)}{\tilde \eta_j} & \text{if }|\tilde g_j| \leq \lambda(\tilde \eta_j+1) \\
   \frac {S(\tilde g_j,\gamma \lambda/ (\gamma-1))}{\tilde \eta_j-1/(\gamma-1)} & \text{if } \lambda(\tilde \eta_j+1) \leq |\tilde g_j| \leq \tilde \eta_j \lambda \gamma \\
   \frac{\tilde g_j}{\tilde \eta_j} & \text{if } |\tilde g_j| \geq \tilde \eta_j \lambda \gamma,
  \end{cases}
\]
for any $\gamma>\max_j(1+1/ \tilde \eta_j)$. Then, for $\gamma>\max_j(1/ \tilde \eta_j)$ and the same definition of $\tilde g_j$ and $\tilde \eta_j$, the coordinate-wise update for MC+ is
\[
 {\tilde \beta_j} \leftarrow 
  \begin{cases} 
     \frac {S(\tilde g_j,\lambda)}{\tilde \eta_j-1/\gamma} & \text{if } |\tilde g_j| \leq \tilde \eta_j \lambda \gamma \\
   \frac{\tilde g_j}{\tilde \eta_j} & \text{if } |\tilde g_j| \geq \tilde \eta_j \lambda \gamma,
  \end{cases}
\]
where $S(z,\lambda)$ is the soft-thresholding operator given by (\ref{eq:sz}).

\subsection{Selection of regularization or tuning parameter}
It is worth noticing that coordinate descent procedures (and other computation procedures computing the penalized likelihood estimates) rely on the tuning parameter $\lambda$ so that the choice of $\lambda$ is also becoming an important task. The estimation using a large value of $\lambda$ tends to have smaller variance but larger biases, whereas the estimation using a small value of $\lambda$ leads to have zero biases but larger variance. The trade-off between the biases and the variances yields an optimal choice of $\lambda$ \citep{fan2010selective}.

To select $\lambda$, it is reasonable to identify a range of $\lambda$ values extending from a maximum value of $\lambda$ for which all penalized coefficients are zero to $\lambda=0$ \cite[e.g.][]{friedman2010regularization, breheny2011coordinate}.  After that, we select a $\lambda$ value which optimizes some criterion. By fixing a path of $\lambda \geq 0$, we select the tuning parameter $\lambda$ which minimizes $\mathrm{WQBIC}(\lambda)$, a weighted version of the BIC criterion, defined by
\begin{align*}
\mathrm{WQBIC}(\lambda)=-2 \ell(w;\boldsymbol{\hat \beta} (\lambda))+ s(\lambda)\log|D|,
\end{align*}
where $s(\lambda)={\sum_{j=1}^p \mathbb{I}\{{\hat \beta_j}(\lambda) \neq 0\}}$ is the number of selected covariates with nonzero regression coefficients and $|D|$ is the observation volume which represents the sample size. For linear regression models, $\mathbf{Y}=\mathbf{X}^\top \boldsymbol{\hat \beta}+\boldsymbol{\epsilon}$, \cite{wang2007tuning} proposed BIC-type criterion for choosing $\lambda$  by 
\begin{align*}
\mathrm{BIC}(\lambda)=\log \frac{\| \mathbf{Y}-\mathbf{X}^\top \boldsymbol{\hat \beta}(\lambda)\|^2}{\eta}+ \frac{1}{\eta}\log (\eta) \mathrm{DF}(\lambda),
\end{align*}
where $\eta$ is the number of observations and $\mathrm{DF}(\lambda)$ is the degree of freedom. This criterion is consistent, meaning that, it selects the correct model with probability approaching 1 in large samples when a set of candidate models contains the true model. Their findings is in line with the study of \cite{zhang2010regularization} for which the criterion was presented in more general way, called generalized information criterion (GIC). The criterion WQBIC is the specific form of GIC proposed by \cite{zhang2010regularization}.

The selection of $\gamma$ for SCAD and MC+ is another task, but we fix $\gamma=3.7$ for SCAD and $\gamma=3$ for MC+, following \cite{fan2001variable} and \cite{breheny2011coordinate} respectively, to avoid more complexities.

\section{Asymptotic theory}
\label{sec:asy}

In this section, we present the asymptotic results for the regularized weighted Poisson likelihood estimator when considering $\mathbf{X}$ as a $d$-dimensional point process observed over a sequence of observation domain $D=D_n, n=1,2, \ldots$ which expands to $\mathbb{R}^d$ as $n \to \infty$. The regularization parameters $\lambda_j=\lambda_{n,j}$ for $j=1,\dots,p$ are now indexed by $n$.
These asymptotic results also hold for the regularized unweighted Poisson likelihood estimator. For sake of conciseness, we do not present the asymptotic results for the regularized logistic regression estimate. The results are very similar. The main difference is lying in the conditions ($\mathcal C$.\ref{C:BnCn}) and ($\mathcal C$.\ref{C:An}) for which the matrices  $\mathbf{A}_n, \mathbf{B}_n,$ and $\mathbf{C}_n$ have a different expression (see Remark \ref{logistic}).

\subsection{Notation and conditions} \label{sec:not}
We recall the classical definition of strong mixing coefficients adapted to spatial point processes \citep[e.g.][]{politis1998large}: for $k,l \in \mathbb{N} \cup \{ \infty \}$ and $q \geq 1$, define
\begin{align}
\alpha_{k,l}(q)=\sup\{
&|\mathrm{P}(A \cap B)-\mathrm{P}(A)\mathrm{P}(B)|: A\in \mathscr{F} (\Lambda_1), B \in \mathscr{F} (\Lambda_2), \nonumber\\
&\Lambda_1 \in \mathscr{B} (\mathbb{R}^d), \Lambda_2 \in \mathscr{B} (\mathbb{R}^d), |\Lambda_1| \leq k, |\Lambda_2| \leq l, d(\Lambda_1, \Lambda_2) \geq q \}, \label{eq:5}
\end{align}
where $\mathscr{F}$ is the $\sigma$-algebra generated by $\mathbf{X} \cap \Lambda_i, i=1,2, d(\Lambda_1, \Lambda_2)$ is the minimal distance between sets $\Lambda_1$ and $\Lambda_2$, and $\mathscr{B}(\mathbb{R}^d)$ denotes the class of Borel sets in $\mathbb{R}^d$.

Let $\boldsymbol \beta_0 = \{\beta_{01},\ldots,\beta_{0s}, \beta_{0(s+1)}, \ldots, \beta_{0p}\}^\top=\{\boldsymbol \beta^{\top}_{01},\boldsymbol \beta^{\top}_{02}\}^\top = (\boldsymbol \beta_{01}^\top, \mathbf 0^\top)^\top$ denote the $p$-dimensional vector of true coefficient values, where ${\boldsymbol \beta_{01}}$ is the $s$-dimensional vector of nonzero coefficients and $\boldsymbol \beta_{02}$ is the ({\em p-s})-dimensional vector of zero coefficients. 

We  define the $p \times p$ matrices $\mathbf{A}_n(w;\boldsymbol{\beta}_{0}), \mathbf{B}_n(w;\boldsymbol{\beta}_{0}),$ and $\mathbf{C}_n(w;\boldsymbol{\beta}_{0})$ by
\begin{align*}
\mathbf{A}_n(w;\boldsymbol{\beta}_{0})&={\int_{D_n} w(u)\mathbf{z}(u)\mathbf{z}(u)^\top \rho(u;\boldsymbol{\beta}_{0})\mathrm{d}u}, \\
\mathbf{B}_n(w;\boldsymbol{\beta}_{0})&={\int_{D_n} w(u)^2 \mathbf{z}(u)\mathbf{z}(u)^\top \rho(u;\boldsymbol{\beta}_{0})\mathrm{d}u}, \mbox{ and} \\
\mathbf{C}_n(w;\boldsymbol{\beta}_{0})&={\int_{D_n} \int_{D_n} w(u) w(v)\mathbf{z}(u)\mathbf{z}(v)^\top \{g(u,v)-1\} \rho(u;\boldsymbol{\beta}_{0}) \rho(v;\boldsymbol{\beta}_{0}) \mathrm{d}v \mathrm{d}u}.
\end{align*}

Consider the following conditions ($\mathcal C$.\ref{C:Dn})-($\mathcal C$.\ref{C:plambda}) which are required to derive our asymptotic results, where $o$ denotes the origin of $\mathbb{R}^d$:

\begin{enumerate}[($\mathcal C$.1)]
\item  For every $n \geq 1, D_n=nE=\{ne: e \in E\}$, where $E \subset \mathbb{R}^d$ is convex, compact, and contains $o$ in its interior. \label{C:Dn}
\item We assume that the intensity function has the log-linear specification given by~\eqref{intensity function} where $\beta \in \Theta$ and $\Theta$ is an open convex bounded set of $\mathbb R^p$. \label{C:Theta}
\item  The covariates $\mathbf z$ and the weight function $w$ satisfy
\[
 	\sup_{u \in \mathbb{R}^d} ||\mathbf{z}(u)||<\infty 
 	\quad \mbox{ and } \quad 
 	\sup_{u \in \mathbb{R}^d} |w(u)|<\infty.
 \] \label{C:cov}
\item  There exists an integer $t \geq 1$ such that for $k=2, \ldots, 2+t$, the product density $\rho^{(k)}$ exists and satisfies $\rho^{(k)}< \infty$. \label{C:rhok}
\item  For the strong mixing coefficients (\ref{eq:5}), we assume that there exists some $\tilde t > d(2+t)/ t$ such that $\alpha_{2,\infty}(q)=O(q^{-\tilde t})$. \label{C:mixing}
\item There exists a $p \times p$ positive definite matrix $\mathbf{I}_0$ such that for all sufficiently large $n$, $|D_n|^{-1}\{\mathbf{B}_n(w;\boldsymbol \beta_0)+\mathbf{C}_n(w;\boldsymbol \beta_0)\}\geq \mathbf{I}_0$. \label{C:BnCn}
\item  There exists a $p \times p$ positive definite matrix $\mathbf{I}_0^\prime$ such that for all sufficiently large $n$, we have $|D_n|^{-1} \mathbf{A}_n(w;\boldsymbol \beta_0)\geq \mathbf{I}_0^\prime$. \label{C:An}
\item   The penalty function $p_{\lambda}(\cdot)$ is nonnegative on $\mathbb R+$, continuously differentiable on $\mathbb R^+ \setminus\{0\}$ with derivative ${p}_\lambda'$ assumed to be a Lipschitz function on $\mathbb R^+\setminus\{0\}$.
Furthermore, given $(\lambda_{n,j})_{n \geq 1}, \mbox{ for } j=1, \ldots, s,$ we assume that there exists $(\tilde r_{n,j})_{n \geq 1}$, where $|D_n|^{1/2}\tilde r_{n,j} \to \infty$ as $n \to \infty$, such that, for $n$ sufficiently large, $p_{\lambda_{n,j}}$ is thrice continuously differentiable in the ball centered at $|\beta_{0j}|$ with radius $\tilde r_{n,j}$ and we assume that the third derivative is uniformly bounded. \label{C:plambda}
\end{enumerate}

Under the condition ($\mathcal C$.\ref{C:plambda}), we define the sequences $a_n$, $b_n$ and $c_n$ by
\begin{align}	
a_n &=\max_{j=1,...s} | p'_{\lambda_{n,j}}(|\beta_{0j}|)| , \label{eq:an} \\
b_n &=\inf_{j=s+1,\ldots,p} \inf_{\substack{|\theta| \leq \epsilon_n \\ \theta \neq 0}} p'_{\lambda_{n,j}}(\theta) \label{eq:bn}, \mbox{ for } \epsilon_n=K_1|D_n|^{-1/2},	\\
c_n &=  \max_{j=1,...s} |p^{\prime\prime}_{\lambda_{n,j}}(|\beta_{0j}|) |.  \label{eq:cn}
\end{align}
These sequences $a_n$, $b_n$ and $c_n$, detailed in Table~\ref{table:2} for the different methods considered in this paper, play a central role in our results. Even if this will be discussed later in Section~\ref{sec:discussion}, we specify right now that we require that $a_n |D_n|^{1/2}\to 0$, $b_n |D_n|^{1/2}\to \infty$ and $c_n\to 0$. 

\setlength{\tabcolsep}{3pt}
\renewcommand{\arraystretch}{1.5}
\begin{table}[h]
\caption{Details of the sequences $a_n$, $b_n$ and $c_n$ for a given regularization method.}
\label{table:2}
\begin{center}
\begin{threeparttable}
\scalebox{0.79}{
\begin{tabular}{ l l l l}
\hline
\hline
Method &$a_n$ & $b_n$ & $c_n$\\
\hline
\hline
Ridge & $\lambda_n {\displaystyle \max_{j=1,...s} \{|\beta_{0j}|\}}$  & $0$ & $\lambda_n$ \\
Lasso & $\lambda_n$ & $\lambda_n$ & 0  \\
Enet & $\lambda_n \left[(1-\gamma){\displaystyle \max_{j=1,...s} \{|\beta_{0j}|\}}+\gamma \right]$  &$\gamma \lambda_n$& $(1-\gamma) \lambda_n$\\
AL & ${\displaystyle \max_{j=1,...s} \{\lambda_{n,j}\}}$ & ${\displaystyle \min_{j=s+1,...p} \{\lambda_{n,j}\}}$ & 0\\ 
Aenet & $ {\displaystyle \max_{j=1,...s}\{ \lambda_{n,j}\big((1-\gamma) |\beta_{0j}|+\gamma\big) \}}$  & $\gamma {\displaystyle \min_{j=s+1,...p}\{\lambda_{n,j} \}}$ & $(1-\gamma){\displaystyle\max_{j=1,\dots,s} \{\lambda_{n,j}\}}$\\ 
SCAD & $0 \tnote{*}$ & $\lambda_n  \tnote{**}$  & $0\tnote{*}$\\ 
MC+ & $0  \tnote{*}$ & $\lambda_n - \frac {K_1} {\gamma |D_n|^{1/2}}  \tnote{**} $ &$0\tnote{*}$\\ 
\hline
\end{tabular}
}
\begin{tablenotes}
\item[*] if $\lambda_n \to 0$ as $n \to \infty$ 
\item[**] if $|D_n|^{1/2}\lambda_n \to \infty$ as $n \to \infty$
\end{tablenotes}
\end{threeparttable}
\end{center}
\end{table}

\subsection{Main results} \label{sec:result}
We state our main results here. Proofs are relegated to Appendices \ref{sec:auxLemma}-\ref{proof2}.

We first show in Theorem~\ref{the1} that the penalized weighted Poisson likelihood estimator converges in probability and exhibits its rate of convergence.

\begin{theorem}
\label{the1}
Assume the conditions ($\mathcal C$.\ref{C:Dn})-($\mathcal C$.\ref{C:plambda}) hold and let $a_n$ and $c_n$ be given by (\ref{eq:an}) and~\eqref{eq:cn}. If $a_n=O(|D_n|^{-1/2})$ and $c_n=o(1)$, then there exists a local maximizer ${\boldsymbol {\hat \beta}}$ of $Q(w;\boldsymbol \beta)$  such that ${\bf \| \boldsymbol {\hat \beta} -\boldsymbol \beta_0\|}=O_\mathrm{P}(|D_n|^{-1/2}+a_n)$.
\end{theorem}

This implies that, if $a_n=O(|D_n|^{-1/2})$ and $c_n=o(1)$, the penalized weighted Poisson likelihood estimator is root-$|D_n|$ consistent. Furthermore, we demonstrate in Theorem~\ref{the2} that such a root-$|D_n|$ consistent estimator ensures the sparsity of $\boldsymbol{\hat \beta}$; that is, the estimate will correctly set $\boldsymbol \beta_2$ to zero with probability tending to 1 as $n \to \infty$, and $\boldsymbol{\hat \beta}_1$ is asymptotically normal.
\begin{theorem}
\label{the2}
Assume the conditions ($\mathcal C$.\ref{C:Dn})-($\mathcal C$.\ref{C:plambda}) hold. If $a_n|D_n|^{1/2}\to 0$, \linebreak$b_n|D_n|^{1/2} \to \infty$ and $c_n\to 0$ as $n\to\infty$, the root-$|D_n|$ consistent local maximizers ${ \boldsymbol {\hat { \beta}}}=(\boldsymbol{\hat \beta}_1^\top, \boldsymbol{\hat \beta}_2^\top)^\top $ in Theorem 1 satisfy:
\begin{enumerate}[(i)]
\item Sparsity: $\mathrm{P}(\boldsymbol{\hat \beta}_2=0) \to 1$ as $n \to \infty$,
\item Asymptotic Normality: $|D_n|^{1/2} \boldsymbol \Sigma_n(w;\boldsymbol{\beta}_{0})^{-1/2}(\boldsymbol{\hat \beta}_1- \boldsymbol{\beta}_{01})\xrightarrow{d} \mathcal{N}(0, \mathbf{I}_{s})$,
\end{enumerate}
where
\begin{align}
\boldsymbol \Sigma_n(w;\boldsymbol{\beta}_{0})= & |D_n|\{\mathbf{A}_{n,11}(w;\boldsymbol{\beta}_{0})+|D_n| \boldsymbol \Pi_n \}^{-1}\{\mathbf{B}_{n,11}(w;\boldsymbol{\beta}_{0})+\mathbf{C}_{n,11}(w;\boldsymbol{\beta}_{0})\}\nonumber \\
& \{\mathbf{A}_{n,11}(w;\boldsymbol{\beta}_{0})+|D_n| \boldsymbol \Pi_n \}^{-1}, \label{eq:Sigman} \\
\boldsymbol \Pi_n = & \mathrm{diag}\{p''_{\lambda_{n,1}}(|\beta_{01}|),\ldots,p''_{\lambda_{n,s}}(|\beta_{0s}|)\}, \label{eq:pi}
\end{align}
and where $\mathbf{A}_{n,11}(w;\boldsymbol{\beta}_{0})$ $(\mbox{resp. } \mathbf{B}_{n,11}(w;\boldsymbol{\beta}_{0}), \mathbf{C}_{n,11}(w;\boldsymbol{\beta}_{0}))$ is the $s \times s$ top-left corner of $\mathbf{A}_{n}(w;\boldsymbol{\beta}_{0})$ $(\mbox{resp. } \mathbf{B}_{n}(w;\boldsymbol{\beta}_{0}), \mathbf{C}_{n}(w;\boldsymbol{\beta}_{0}))$.
\end{theorem}

As a consequence, $\boldsymbol \Sigma_n(w;\boldsymbol{\beta}_{0})$ is the asymptotic covariance matrix of $\boldsymbol{\hat \beta}_1$. Note that $\boldsymbol \Sigma_n(w;\boldsymbol{\beta}_{0})^{-1/2}$ is the inverse of $\boldsymbol \Sigma_n(w;\boldsymbol{\beta}_{0})^{1/2}$, where $\boldsymbol \Sigma_n(w;\boldsymbol{\beta}_{0})^{1/2}$ is any square matrix with $\boldsymbol \Sigma_n(w;\boldsymbol{\beta}_{0})^{1/2}\big(\boldsymbol \Sigma_n(w;\boldsymbol{\beta}_{0})^{1/2}\big)^\top=\boldsymbol \Sigma_n(w;\boldsymbol{\beta}_{0})$.

\begin{remark}
\label{remark 1}
For lasso and adaptive lasso, $ \boldsymbol \Pi_n= \mathbf 0$. For other penalties, since $c_n=o(1)$, then $\boldsymbol \| \Pi_n\|=o(1)$. Since $\|\mathbf{A}_{n,11}(w;\boldsymbol{\beta}_{0})\|=O(|D_n|)$ from conditions ($\mathcal C$.\ref{C:Theta}) and ($\mathcal C$.\ref{C:cov}), $|D_n| \, \|  \boldsymbol \Pi_n\|$ is asymptotically negligible with respect to $\|\mathbf{A}_{n,11}(w;\boldsymbol{\beta}_{0})\|$.
\end{remark}

\begin{remark}
\label{logistic}
Theorems \ref{the1} and \ref{the2} remain true for the regularized weighted logistic regression likelihood estimates if we extend the condition ($\mathcal C$.\ref{C:cov}) by replacing in the expression of the matrices $\mathbf{A}_n, \mathbf{B}_n,$ and $\mathbf{C}_n$, $w(u)$ by ${w(u) \delta(u)}/ ({\rho(u; \boldsymbol \beta) + \delta(u)}), u \in D_n$ and by adding $\sup_{u \in \mathbb{R}^d} \delta(u)<\infty$.
\end{remark}

\begin{remark}
\label{HL}
We want to highlight here the main theoretical differences with the work by \cite{thurman2015regularized}. First, the methodology and results are available for the logistic regression likelihood. Second, we consider very general penalty function while \cite{thurman2015regularized} only considered the adaptive lasso method. Third, we do not assume, as in \cite{thurman2015regularized}, that $|D_n|^{-1} \mathbf{M}_n \to \mathbf{M}$ as $ n \to \infty$ (where $\mathbf{M}_n$ is $\mathbf{A}_n, \mathbf{B}_n,$ or $\mathbf{C}_n$), when $\mathbf{M}$ is a positive definite matrix. Instead we assume sharper condition assuming $\lim_{n \to \infty}\nu_{\min}(|D_n|^{-1} \mathbf{M}_n)> 0$, where $\mathbf{M}_n$ is either $\mathbf{A}_n$ or $ \mathbf{B}_n+\mathbf{C}_n$ and $\nu_{\min}(\mathbf{M'})$ is the smallest eigenvalue of a positive definite matrix $\mathbf{M'}$. This makes the proofs a little bit more technical.
\end{remark}

\subsection{Discussion of the conditions} \label{sec:discussion}
We adopt the conditions ($\mathcal C$.\ref{C:Dn})-($\mathcal C$.\ref{C:BnCn}) based on the paper from \cite{coeurjolly2014variational}. In condition ($\mathcal C$.\ref{C:Dn}), the assumption that $E$ contains $o$ in its interior can be made without loss of generality.  If instead $u$ is an interior point of $E$, then condition ($\mathcal C$.\ref{C:Dn}) could be modified to that any ball with centre $u$ and radius $r>0$ is contained in $D_n=nE$ for all sufficiently large $n$. Condition ($\mathcal C$.\ref{C:cov}) is quite standard. From conditions ($\mathcal C$.\ref{C:Theta})-($\mathcal C$.\ref{C:mixing}), the matrices $\mathbf{A}_n(w;\boldsymbol \beta_0)$, $\mathbf{B}_n(w;\boldsymbol \beta_0)$ and $\mathbf{C}_n(w;\boldsymbol \beta_0)$ are bounded by $|D_n|$ \citep[see e.g.][]{coeurjolly2014variational}.

Combination of conditions ($\mathcal C$.\ref{C:Dn})-($\mathcal C$.\ref{C:BnCn}) are used to establish a central limit theorem for $|D_n|^{-1/2}\ell_n^{(1)}(w;\boldsymbol \beta_0)$ using a general central limit theorem for triangular arrays of nonstationary random fields obtained by \cite{karaczony2006central}, which is an extension from \cite{bolthausen1982central}, then later extended to nonstationary random fields by \cite{guyon1995random}. As pointed out by \cite{coeurjolly2014variational}, condition ($\mathcal C$.\ref{C:BnCn}) is a spatial average assumption like when establishing asymptotic normality of ordinary least square estimators for linear models. This condition is also useful to make sure that the matrix $|D_n|^{-1}\{\mathbf{B}_n(w;\boldsymbol{\beta}_{0})+\mathbf{C}_n(w;\boldsymbol{\beta}_{0})\}$ is invertible. 
Conditions ($\mathcal C$.\ref{C:BnCn})-($\mathcal C$.\ref{C:An}) ensure that the matrix $\boldsymbol \Sigma_n(w;\boldsymbol{\beta}_{0})$ is invertible for sufficiently large $n$. Conditions ($\mathcal C$.\ref{C:Dn})-($\mathcal C$.\ref{C:BnCn}) are discussed in details for several models by \cite{coeurjolly2014variational}. They are satisfied for a large class of intensity functions and a large class of models including Poisson and Cox processes discussed in Section \ref{sec:intensity}.

Condition ($\mathcal C$.\ref{C:plambda}) controls the higher order terms in Taylor expansion of the penalty function. Roughly speaking, we ask the penalty function to be at least Lipschitz and thrice differentiable in a neighborhood of the true parameter vector. As it is, the condition looks technical, however, 
it is obviously satisfied for ridge, lasso, elastic net (and the adaptive versions). According to the choice of $\lambda_n$, it is satisfied for SCAD and MC+ when $|\beta_{0j}|$, for $j=1,\ldots,s$, is not equal to $\gamma \lambda_n$ and/or $\lambda_n$.

Theorem \ref{the2} requires the conditions $a_n|D_n|^{1/2} \to 0$, $b_n|D_n|^{1/2} \to \infty$ and $c_n\to 0$ as $n \to \infty$ simultaneously. By requiring these assumptions, the corresponding penalized weighted Poisson likelihood estimators possess the oracle property and perform as well as weighted Poisson likelihood estimator which estimates $\boldsymbol{\beta}_{1}$ knowing the fact that $\boldsymbol{\beta}_{2}=\mathbf{0}$.

For the ridge regularization method, $b_n=0$, preventing from applying Theorem~\ref{the2} for this penalty. For lasso and elastic net, $a_n=K_2 b_n$ for some constant $K_2>0$ ($K_2$=1 for lasso). The two conditions $a_n|D_n|^{1/2} \to 0$ and $b_n|D_n|^{1/2} \to \infty$ as $n \to \infty$ cannot be satisfied simultaneously. This is different for the adaptive versions where a compromise can be found by adjusting the $\lambda_{n,j}$'s, as well as the two non-convex penalties SCAD and MC+, for which $\lambda_n$ can be adjusted. For the regularization methods considered in this paper, the condition $c_n\to 0$ is implied by the condition $a_n|D_n|^{1/2}\to 0$ as $n\to \infty$.

\section{Simulation study}
\label{simul}
We conduct a simulation study with three different scenarios, described in Section \ref{setup}, to compare the estimates of the regularized Poisson likelihood (PL) and that of the regularized weighted Poisson likelihood (WPL). We also want to explore the behaviour of the estimates using different regularization methods. Empirical findings are presented in Section \ref{result}. Furthermore, we compare, in Section \ref{sec:logi}, the estimates of the regularized (un)weighted logistic likelihood and the ones of the regularized (un)weighted Poisson likelihood.

\subsection{Simulation set-up}
\label{setup}
The setting is quite similar to that of \cite{waagepetersen2007estimating} and \cite{thurman2015regularized}. The spatial domain is $D=[0,1000] \times [0,500]$. We center and scale the $201 \times 101$ pixel images of elevation ($x_1$) and gradient of elevation ($x_2$) contained in the $\texttt{bei}$ datasets of $\texttt{spatstat}$ library in $\texttt{R}$ \citep{team2014r}, and use them as two true covariates. In addition, we create three different scenarios to define extra covariates:

\begin{enumerate}[{Scenario}~1.]
\item We generate eighteen $201 \times 101$ pixel images of covariates as standard Gaussian white noise and denote them by $x_3, \ldots, x_{20}$. We define $\mathbf{z}(u)=\mathbf{x}(u)=\{x_1(u),\ldots,x_{20}(u)\}^\top$ as the covariates vector. The regression coefficients for $z_3, \ldots, z_{20}$ are set to zero. \label{sce1}
\item First, we generate eighteen $201 \times 101$ pixel images of covariates as in the scenario 1. Second, we transform them, together with $x_1$ and $x_2$, to have multicollinearity. Third, we define $\mathbf{z}(u)=\mathbf{V} ^\top \mathbf{x}(u)$, where $\mathbf{x}(u)=\{x_1(u),\ldots,x_{20}(u)\}^\top$. More precisely, $\mathbf{V}$ is such that $\boldsymbol \Omega=\mathbf{V} ^\top \mathbf{V},$ and $(\Omega)_{ij}=(\Omega)_{ji}=0.7^{|i-j|}$ for $i,j=1,\ldots,20$, except $(\Omega)_{12}=(\Omega)_{21}=0$, to preserve the correlation between $x_1$ and $x_2$. The regression coefficients for $z_3, \ldots, z_{20}$ are set to zero. \label{sce2}
\item We consider a more complex situation. We center and scale the 13 soil nutrients covariates obtained from the study in tropical forest of Barro Colorado Island (BCI) in central Panama \citep[see][]{condit1998tropical, hubbell1999light, hubbell2005barro}, and use them as the extra covariates. Together with $x_1$ and $x_2$, we keep the structure of the covariance matrix to preserve the complexity of the situation. In this setting, we have $\mathbf{z}(u)=\mathbf{x}(u)=\{x_1(u),\ldots,x_{15}(u)\}^\top$.  The regression coefficients for $z_3, \ldots, z_{15}$ are set to zero. \label{sce3}
\end{enumerate}

The different maps of the covariates obtained from scenarios~\ref{sce2} and \ref{sce3} are depicted in Appendix \ref{mapcov}. Except for $z_3$ which has high correlation with $z_2$, the extra covariates obtained from scenario~\ref{sce2} tend to have a constant value (Figure~\ref{sim2}). This is completely different from the ones obtained from scenario~\ref{sce3} (Figure~\ref{sim3}).

The mean number of points over the domain $D$, $\mu$, is chosen to be 1600. We set the true intensity function to be $\rho(u; \boldsymbol \beta_0)=\{\beta_0+\beta_1 z_1(u)+\beta_2 z_2(u)\}$, where $\beta_1=2$ represents a relatively large effect of elevation, $\beta_2=0.75$ reflects a relatively small effect of gradient, and $\beta_0$ is selected such that each realization has 1600 points in average. Furthermore, we erode regularly the domain $D$ such that, with the same intensity function, the mean number of points over the new domain $D \ominus R$ becomes 400. The erosion is used to observe the convergence of the procedure as the observation domain expands. We consider the default number of dummy points for the Poisson likelihood, denoted by $\texttt{nd}^2$, as suggested in the $\texttt{spatstat}$ $\texttt{R}$ package, i.e. $\texttt{nd}^2 \approx 4m$, where $m$ is the number of points. With these scenarios, we simulate 2000 spatial point patterns from a Thomas point process using the $\mathtt{rThomas}$ function in the $\mathtt{spatstat}$ package. We also consider two different $\kappa$ parameters $(\kappa=5 \times 10^{-4}, \kappa=5 \times 10^{-5})$ as different levels of spatial interaction and let $\omega=20$. For each of the four combinations of $\kappa$ and $\mu$, we fit the intensity to the simulated point pattern realizations. We also fit the oracle model which only uses the two true covariates.

All models are fitted using modified internal function in $\mathtt{spatstat}$ \citep{baddeley2015spatial}, $\mathtt{glmnet}$ \citep{friedman2010regularization}, and $\mathtt{ncvreg}$ \citep{breheny2011coordinate}. A modification of the $\mathtt{ncvreg}$ $\mathtt{R}$ package is required to include the penalized weighted Poisson and logistic likelihood methods.

\subsection{Simulation results}
\label{result}
To better understand the behaviour of Thomas processes designed in this study, Figure \ref{fig:plot} shows the plot of the four realizations using different $\kappa$ and $\mu$. The smaller value of $\kappa$, the tighter the clusters since there are fewer parents. When $\mu=400$, i.e. by considering the realizations observed on $D \ominus R$, the mean number of points over the 2000 replications and standard deviation are 396 and 47 (resp. 400 and 137) when $\kappa=5 \times 10^{-4}$ (resp. $\kappa=5 \times 10^{-5}$). When $\mu=1600$, the mean number of points and standard deviation are 1604 and 174 (resp. 1589 and 529) when $\kappa=5 \times 10^{-4}$ (resp. $\kappa=5 \times 10^{-5}$).

\begin{figure}[!ht]
\begin{center}
\graphicspath{{d:/Figure/}}
\setlength{\tabcolsep}{0pt}
\renewcommand{\arraystretch}{0}
\begin{tabular}{c c}
\includegraphics[width=0.477\textwidth, height=0.11\textheight]{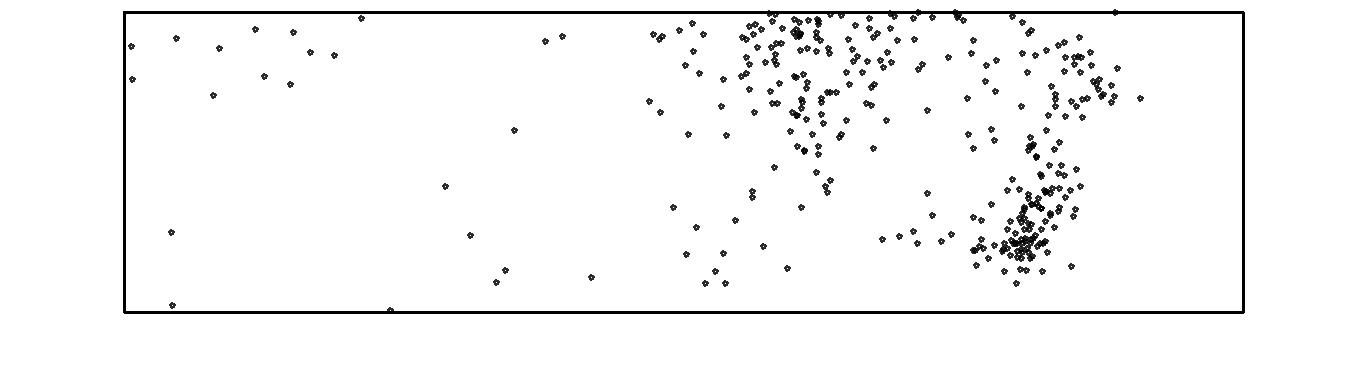} & \includegraphics[width=0.477\textwidth, height=0.11\textheight]{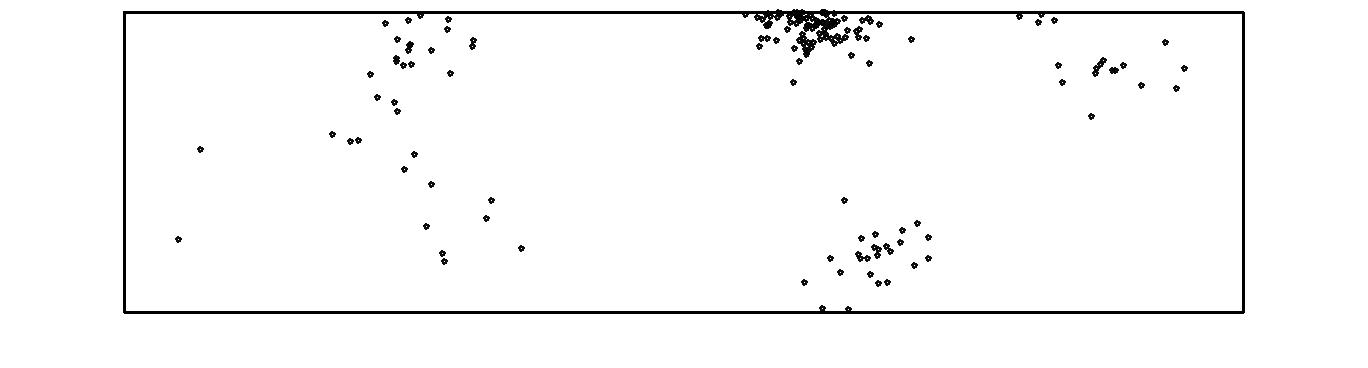}\\
\includegraphics[width=0.45\textwidth, height=0.1\textheight]{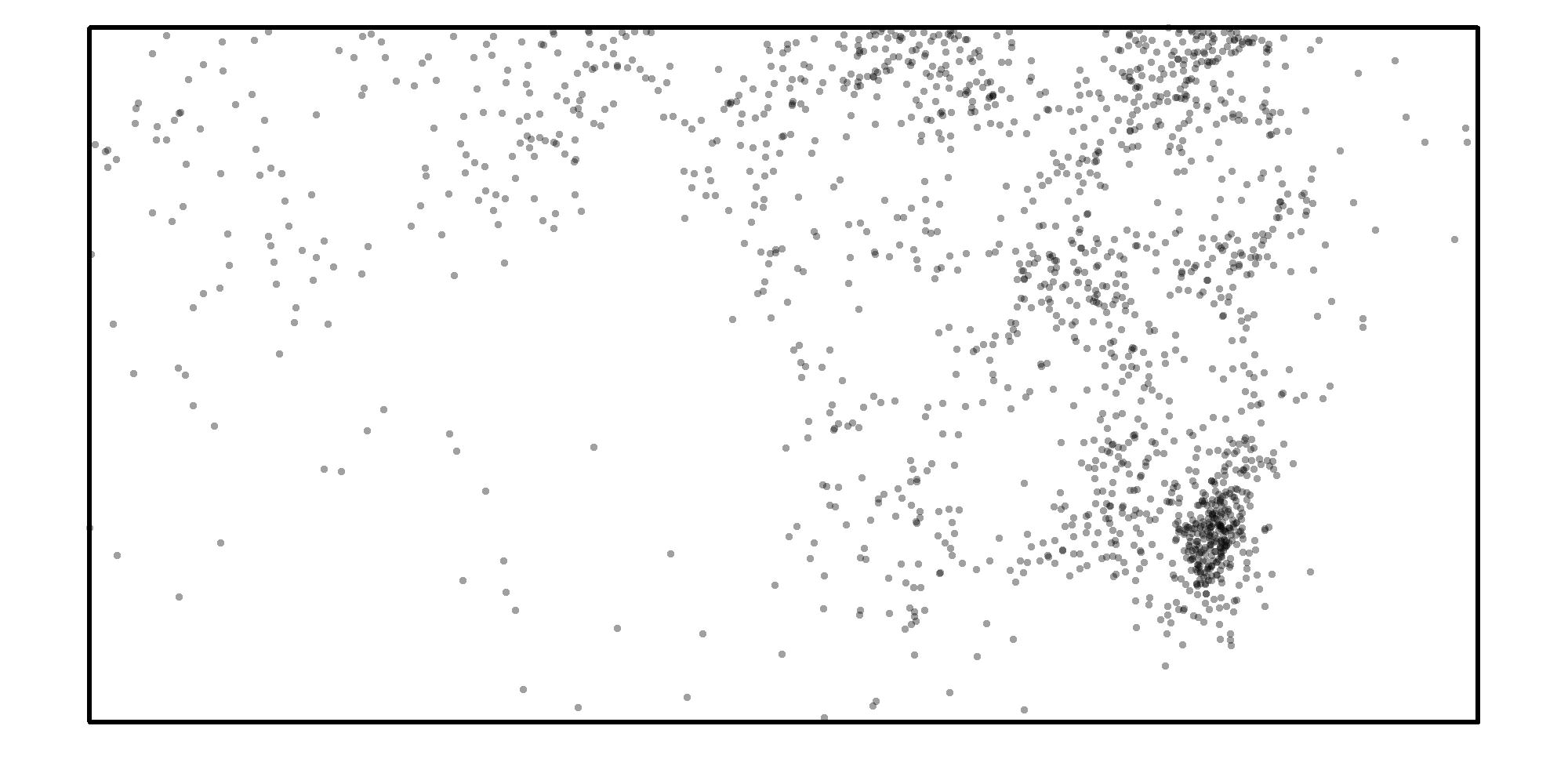} & \includegraphics[width=0.45\textwidth, height=0.1\textheight]{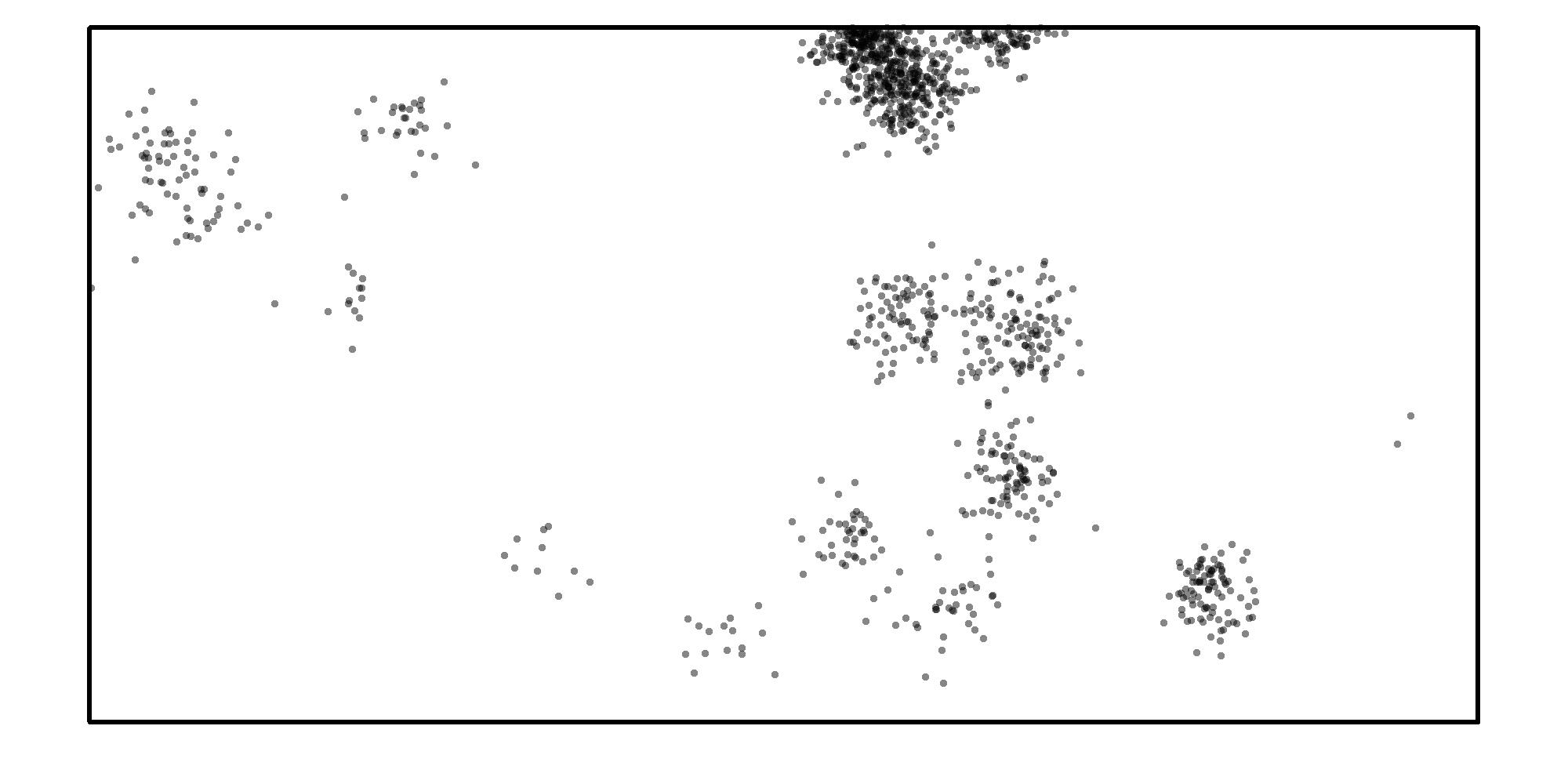}\\
\end{tabular}
\caption{Realizations of a Thomas process for $\mu=400$ (row 1), $\mu=1600$ (row 2), $\kappa=5 \times 10^{-4}$ (column 1), and $\kappa=5 \times 10^{-5}$ (column 2). }
\label{fig:plot}
\end{center}
\end{figure}

\setlength{\tabcolsep}{3pt}
\renewcommand{\arraystretch}{1.5}
\begin{table}[!ht]
\caption{Empirical selection properties (TPR, FPR, and PPV in $\%$) based on 2000 replications of Thomas processes on the domain $D \ominus R$ ($\mu=400$) for different values of $\kappa$ and for the three different scenarios. Different penalty functions are considered as well as two estimating equations, the regularized Poisson likelihood (PL) and the regularized weighted Poisson likelihood (WPL).}
\label{table:rate}
\centering
\begin{threeparttable}
\scalebox{0.79}{
\begin{tabular}{ l ccc | ccc | ccc | ccc}
\hline
\hline
\multirow{3}{*}{Method} & \multicolumn{6}{c}{$\kappa=5 \times 10^{-4}$}  & \multicolumn{6}{c}{$\kappa=5 \times 10^{-5}$}\\
\cline{2-13}
 & \multicolumn{3}{c}{Regularized PL} & \multicolumn{3}{c}{Regularized WPL} & \multicolumn{3}{c}{Regularized PL} & \multicolumn{3}{c}{Regularized WPL} \\
\cline{2-13}
 & TPR & FPR & PPV & TPR & FPR & PPV & TPR & FPR & PPV & TPR & FPR & PPV \\ 
  \hline
\hline
& \multicolumn{12}{c}{Scenario~\ref{sce1}}\\
\hline
  Ridge & 100 & 100 & 10 & 100 & 100 & 10 & 100 & 100 & 10 & 100 & 100 & 10 \\ 
  Lasso & 100\tnote{*} & 27 & 35 & 56 & 0\tnote{*} & 98 &  89 & 35 & 34 & 33 & 0\tnote{*} & 62 \\ 
  Enet & 100\tnote{*} & 59 & 18 & 39 & 4 & 36 & 91 & 60 & 21 & 31 & 0\tnote{*} & 57 \\ 
  AL & 100\tnote{*} & 1 & 93 & 58 & 0\tnote{*} & 100\tnote{*} &  88 & 7 & 72 & 35 & 0\tnote{*} & 67 \\ 
  Aenet & 100\tnote{*} & 6 & 72 & 59 & 0\tnote{*} & 99 &  89 & 12 & 61 & 34 & 0\tnote{*} & 64 \\ 
  SCAD & 100\tnote{*} & 18 & 41 & 66 & 0\tnote{*} & 98 & 90 & 17 & 46 & 31 & 0\tnote{*} & 56 \\ 
  MC+ & 100\tnote{*} & 21 & 36 & 68 & 0\tnote{*} & 96 & 90 & 21 & 42 & 30 & 0\tnote{*} & 54\\
   \hline
& \multicolumn{12}{c}{Scenario~\ref{sce2}}\\
\hline
  Ridge & 100 & 100 & 10 & 100 & 100 & 10 & 100 & 100 & 10 & 100 & 100 & 10 \\ 
    Lasso & 100\tnote{*} & 25 & 35 & 52 & 1 & 88 & 90 & 38 & 29 & 31 & 0\tnote{*} & 55\\ 
  Enet & 100\tnote{*} & 52 & 19 & 49 & 4 & 62 & 90 & 60 & 20 & 24 & 1 & 38 \\ 
  AL & 99 & 4 & 80 & 52 & 0\tnote{*} & 100\tnote{*} & 87 & 9 & 67 & 36 & 0\tnote{*} & 67 \\ 
  Aenet & 99 & 8 & 65 & 53 & 0\tnote{*} & 99 & 88 & 14 & 54 & 35 & 0\tnote{*} & 65 \\ 
  SCAD & 100\tnote{*} & 17 & 43 & 64 & 0\tnote{*} & 92 & 88 & 17 & 45 & 28 & 0\tnote{*} & 50  \\ 
  MC+ & 100\tnote{*} & 18 & 41 & 59 & 1 & 87 & 88 & 21 & 41 & 27 & 0\tnote{*} & 50\\ 
   \hline
& \multicolumn{12}{c}{Scenario~\ref{sce3}}\\
\hline
  Ridge & 100 & 100 & 13 & 100 & 100 & 13 & 100 & 100 & 13 & 100 & 100 & 13 \\
  Lasso & 100\tnote{*} & 56 & 24 & 52 & 2 & 87  & 98 & 89 & 15 & 13 & 2 & 20 \\ 
  Enet & 100\tnote{*} & 76 & 18 & 47 & 4 & 63 & 99 & 94 & 14 & 8 & 2 & 11 \\ 
  AL & 100\tnote{*} & 29 & 42 & 52 & 0\tnote{*} & 100\tnote{*} & 95 & 77 & 17 & 18 & 2 & 30 \\ 
  Aenet & 100\tnote{*} & 38 & 33 & 54 & 0\tnote{*} & 99 & 96 & 82 & 16 & 15 & 1 & 25 \\ 
  SCAD & 100\tnote{*} & 34 & 33 & 58 & 0\tnote{*} & 85 & 95 & 71 & 18 & 13 & 1 & 22 \\ 
  MC+ & 100\tnote{*} & 35 & 32 & 56 & 0\tnote{*} & 84 & 95 & 71 & 18 & 13 & 1 & 23 \\ 
   \hline

\end{tabular}
}
\begin{tablenotes}
\item[*] Approximate value
\end{tablenotes}
\end{threeparttable}
\end{table}

\setlength{\tabcolsep}{3pt}
\renewcommand{\arraystretch}{1.5}
\begin{table}[!ht]
\caption{Empirical selection properties (TPR, FPR, and PPV in $\%$) based on 2000 replications of Thomas processes on the domain $D$ ($\mu=1600$) for different values of $\kappa$ and for the three different scenarios. Different penalty functions are considered as well as two estimating equations, the regularized Poisson likelihood (PL) and the regularized weighted Poisson likelihood (WPL).}
\label{table:rate2}
\centering
\begin{threeparttable}
\scalebox{0.79}{
\begin{tabular}{ l ccc | ccc | ccc | ccc}
\hline
\hline
\multirow{3}{*}{Method} & \multicolumn{6}{c}{$\kappa=5 \times 10^{-4}$}  & \multicolumn{6}{c}{$\kappa=5 \times 10^{-5}$}\\
\cline{2-13}
 & \multicolumn{3}{c}{Regularized PL} & \multicolumn{3}{c}{Regularized WPL} & \multicolumn{3}{c}{Regularized PL} & \multicolumn{3}{c}{Regularized WPL} \\
\cline{2-13}
 & TPR & FPR & PPV & TPR & FPR & PPV & TPR & FPR & PPV & TPR & FPR & PPV \\ 
  \hline
\hline
& \multicolumn{12}{c}{Scenario~\ref{sce1}}\\
\hline
Ridge & 100 & 100 & 10 & 100 & 100 & 10 & 100 & 100 & 10 & 100 & 100 & 10  \\ 
  Lasso & 100 & 26 & 35 & 52 & 0\tnote{*} & 100\tnote{*}  & 98 & 48 & 22 & 56 & 0\tnote{*} & 96\\ 
  Enet & 100 & 64 & 16 & 55 & 6 & 50 & 99 & 76 & 14 & 50 & 5 & 45 \\ 
  AL & 100 & 0\tnote{*} & 98 & 50 & 0 & 100 & 96 & 6 & 77 & 55 & 0\tnote{*} & 98 \\ 
  Aenet & 100 & 4 & 79 & 54 & 0\tnote{*} & 100\tnote{*} & 97 & 11 & 60 & 57 & 0\tnote{*} & 96 \\ 
  SCAD & 100 & 17 & 50 & 60 & 0\tnote{*} & 100\tnote{*} & 98 & 18 & 47 & 52 & 0\tnote{*} & 90 \\ 
  MC+ & 100 & 22 & 47 & 60 & 0\tnote{*} & 97 & 98 & 23 & 42 & 44 & 0\tnote{*} & 79 \\ 
\hline
& \multicolumn{12}{c}{Scenario~\ref{sce2}}\\
\hline
Ridge & 100 & 100 & 10 & 100 & 100 & 10  & 100 & 100 & 10 & 100 & 100 & 10 \\ 
  Lasso & 100 & 26 & 33 & 51 & 0\tnote{*} & 97  &  98 & 43 & 24 & 52 & 1 & 91 \\ 
  Enet & 100 & 56 & 18 & 51 & 5 & 55  &  99 & 69 & 15 & 49 & 4 & 62 \\ 
  AL & 100 & 1 & 92 & 51 & 0 & 100  & 96 & 10 & 67 & 53 & 0\tnote{*} & 99\\ 
  Aenet  & 100 & 4 & 78 & 51 & 0\tnote{*} & 100\tnote{*}  &  97 & 15 & 52 & 53 & 0\tnote{*} & 98\\ 
  SCAD & 100 & 21 & 37 & 53 & 1 & 85  & 96 & 16 & 50 & 45 & 1 & 77\\ 
  MC+ & 100 & 24 & 35 & 47 & 2 & 76  & 97 & 19 & 47 & 42 & 2 & 72  \\ 
   \hline
& \multicolumn{12}{c}{Scenario~\ref{sce3}}\\
\hline
  Ridge & 100 & 100 & 13 & 100 & 100 & 13 & 100 & 100 & 13 & 100 & 100 & 13 \\ 
  Lasso & 100 & 69 & 19 & 52 & 1 & 96 & 100 & 95 & 14 & 48 & 4 & 75 \\ 
 Enet & 100 & 85 & 16 & 52 & 5 & 71 & 100 & 97 & 14 & 43 & 5 & 62 \\ 
  AL & 100 & 43 & 32 & 51 & 0\tnote{*} & 100\tnote{*} & 99 & 86 & 15 & 51 & 2 & 86 \\ 
  Aenet & 100 & 49 & 27 & 52 & 0\tnote{*} & 99 & 99 & 89 & 15 & 50 & 3 & 82 \\ 
  SCAD & 100 & 47 & 27 & 43 & 2 & 72  & 99 & 78 & 17 & 40 & 2 & 63 \\ 
  MC+ & 100 & 48 & 26 & 44 & 2 & 75& 99 & 79 & 17 & 37 & 2 & 61 \\ 
   \hline

\end{tabular}
}
\begin{tablenotes}
\item[*] Approximate value
\end{tablenotes}
\end{threeparttable}
\end{table}

Tables \ref{table:rate} and \ref{table:rate2} present the selection properties of the estimates using the penalized PL and the penalized WPL methods. Similarly to \cite{buhlmann2011statistics}, the indices we consider are the true positive rate (TPR), the false positive rate (FPR), and the positive predictive value (PPV). TPR corresponds to the ratio of the selected true covariates over the number of true covariates, while FPR corresponds to the ratio of the selected noisy covariates over the number of noisy covariates. TPR explains how the model can correctly select both $z_1$ and $z_2$. Finally, FPR investigates how the model uncorrectly select among $z_3$ to $z_p$ ($p=20$ for scenarios~\ref{sce1} and \ref{sce2} and $p=15$ for scenario~\ref{sce3}). PPV corresponds to the ratio of the selected true covariates over the total number of selected covariates in the model. PPV describes how the model can approximate the oracle model in terms of selection. Therefore, we want to find the methods which have a TPR and a PPV close to 100$\%$, and a FPR close to 0.

Generally, for both the penalized PL and the penalized WPL methods, the best selection properties are obtained for a larger value of $\kappa$ which shows weaker spatial dependence. For a more clustered one, indicated by a smaller value of $\kappa$, it seems more difficult to select the true covariates. As $\mu$ increases from 400 (Table \ref{table:rate}) to 1600 (Table \ref{table:rate2}), the TPR tends to improve, so the model can select both $z_1$ and $z_2$ more frequently.

Ridge, lasso, and elastic net are the regularization methods that cannot satisfy our theorems. It is firstly emphasized that all covariates are always selected by the ridge so that the rates are never changed whatever method used. For the penalized PL with lasso and elastic net regularization, it is shown that they tend to have quite large value of FPR, meaning that they wrongly keep the noisy covariates more frequently. When the penalized WPL is applied, we gain smaller FPR, but we suffer from smaller TPR at the same time. This smaller TPR actually comes from the unselection of $z_2$ which has smaller coefficient than that of $z_1$.

When we apply adaptive lasso, adaptive elastic net, SCAD, and MC+, we achieve better performance, especially for FPR which is closer to zero which automatically improves the PPV. Adaptive elastic net (resp. elastic net) has slightly larger FPR than adaptive lasso (resp. lasso). Among all regularization methods considered in this paper, adaptive lasso seems to outperform the other ones.

Considering scenarios~\ref{sce1} and \ref{sce2}, we observe best selection properties for the penalized PL combined with adaptive lasso. As the design is getting more complex for scenario~\ref{sce3}, applying the penalized PL suffers from much larger FPR, indicating that this method may not be able to overcome the complicated situation. However, when we use the penalized WPL, the properties seem to be more stable for the different designs of simulation study. One more advantage when considering the penalized WPL is that we can remove almost all extra covariates. It is worth noticing that we may suffer from smaller TPR when we apply the penalized WPL, but we lose the only less informative covariates. From Tables \ref{table:rate} and \ref{table:rate2}, when we are faced with complex situation, we would recommend the use of the penalized WPL method with adaptive lasso penalty if the focus is on selection properties. Otherwise, the use of the penalized PL combined with adaptive lasso penalty is more preferable.

\setlength{\tabcolsep}{3pt}
\renewcommand{\arraystretch}{1.5}
\begin{table}[!ht]
\caption{Empirical prediction properties (Bias, SD, and RMSE) based on 2000 replications of Thomas processes on the domain $D \ominus R$ ($\mu=400$) for different values of $\kappa$ and for the three different scenarios. Different penalty functions are considered as well as two estimating equations, the regularized Poisson likelihood (PL) and the regularized weighted Poisson likelihood (WPL).}
\label{table:mean}
\begin{center}
\scalebox{0.79}{
\begin{tabular}{l ccc | ccc |ccc | ccc}
  \hline
\hline
\multirow{3}{*}{Method} & \multicolumn{6}{c}{$\kappa=5 \times 10^{-4}$}  & \multicolumn{6}{c}{$\kappa=5 \times 10^{-5}$}\\
\cline{2-13}
 & \multicolumn{3}{c}{Regularized PL} & \multicolumn{3}{c}{Regularized WPL} & \multicolumn{3}{c}{Regularized PL} & \multicolumn{3}{c}{Regularized WPL} \\
\cline{2-13}
& Bias & SD & RMSE & Bias & SD & RMSE & Bias & SD & RMSE & Bias & SD & RMSE \\ 
  \hline
\hline
& \multicolumn{12}{c}{Scenario~\ref{sce1}}\\
\hline
Oracle & 0.11 & 0.18 & 0.21 & 0.64 & 0.20 & 0.67 & 0.29 & 0.81 & 0.86 & 0.57 & 0.54 & 0.78 \\ 
  Ridge & 0.11 & 0.38 & 0.40 & 0.72 & 0.69 & 1.00 & 0.28 & 1.26 & 1.29 & 0.98 & 1.03 & 1.42 \\ 
  Lasso & 0.28 & 0.32 & 0.42 & 1.06 & 0.32 & 1.11 & 0.47 & 0.99 & 1.10 & 1.40 & 0.73 & 1.58 \\ 
  Enet & 0.24 & 0.38 & 0.44 & 1.28 & 0.28 & 1.31 & 0.45 & 1.04 & 1.13 & 1.59 & 0.58 & 1.70 \\ 
  AL & 0.10 & 0.29 & 0.31 & 0.87 & 0.32 & 0.92 & 0.38 & 0.96 & 1.03 & 1.18 & 0.93 & 1.50 \\ 
  Aenet & 0.14 & 0.30 & 0.33 & 0.93 & 0.39 & 1.01 & 0.40 & 0.96 & 1.04 & 1.29 & 0.82 & 1.53 \\ 
  SCAD & 0.26 & 0.27 & 0.38 & 1.06 & 0.37 & 1.12 & 0.46 & 0.79 & 0.91 & 1.49 & 0.67 & 1.64 \\ 
  MC+ & 0.28 & 0.28 & 0.39 & 1.04 & 0.38 & 1.11 & 0.47 & 0.78 & 0.92 & 1.48 & 0.70 & 1.64 \\ 
\hline
& \multicolumn{12}{c}{Scenario~\ref{sce2}}\\
  \hline
Oracle & 0.12 & 0.23 & 0.26 & 0.71 & 0.26 & 0.76 & 0.30 & 0.78  & 0.84 & 0.59 & 0.62 & 0.84 \\ 
  Ridge & 0.14 & 0.46 & 0.48 & 0.69 & 0.93 & 1.16 & 0.32 & 1.23 & 1.27 & 0.92 & 1.15 & 1.47 \\ 
  Lasso & 0.34 & 0.33 & 0.48 & 1.20 & 0.37 & 1.26 & 0.45 & 0.96 & 1.06 & 1.50 & 0.69 & 1.65 \\ 
  Enet & 0.38 & 0.40 & 0.55 & 1.40 & 0.35 & 1.44 & 0.44 & 1.03 & 1.12 & 1.78 & 0.49 & 1.85 \\ 
  AL & 0.20 & 0.33 & 0.39 & 0.85 & 0.32 & 0.91 & 0.37 & 0.93 & 1.00 & 1.17 & 0.86 & 1.45 \\ 
  Aenet & 0.25 & 0.33 & 0.42 & 0.96 & 0.34 & 1.02 & 0.40 & 0.94 & 1.02 & 1.29 & 0.78 & 1.51 \\ 
  SCAD & 0.38 & 0.30 & 0.48 & 0.95 & 0.48 & 1.06 & 0.44 & 0.80 & 0.91 & 1.53 & 0.70 & 1.68 \\ 
  MC+ & 0.39 & 0.30 & 0.49 & 1.01 & 0.49 & 1.13 & 0.44 & 0.80 & 0.92 & 1.52 & 0.71 & 1.68 \\ 
   \hline
 & \multicolumn{12}{c}{Scenario~\ref{sce3}}\\
\hline
Oracle & 0.12 & 0.46 & 0.48 & 0.70 & 0.26 & 0.75 & 0.65 & 1.14 & 1.31 & 0.87 & 0.88 & 1.24 \\ 
  Ridge & 0.13 & 1.03 & 1.04 & 0.71 & 1.45 & 1.62 & 0.52 & 3.10 & 3.14 & 0.90 & 2.86 & 3.00 \\ 
  Lasso & 0.20 & 0.69 & 0.71 & 1.26 & 0.40 & 1.32 & 0.51 & 2.91 & 2.95 & 1.93 & 0.68 & 2.04 \\ 
  Enet & 0.21 & 0.83 & 0.86 & 1.53 & 0.40 & 1.58 & 0.52 & 2.94 & 2.99 & 2.03 & 0.60 & 2.12 \\ 
  AL & 0.18 & 0.57 & 0.60 & 0.91 & 0.33 & 0.97 & 0.52 & 2.80 & 2.85 & 1.77 & 0.84 & 1.96 \\ 
  Aenet & 0.22 & 0.61 & 0.65 & 1.04 & 0.36 & 1.10 & 0.52 & 2.80 & 2.85 & 1.86 & 0.73 & 2.00 \\ 
  SCAD & 0.27 & 0.61 & 0.67 & 1.18 & 0.59 & 1.32 & 0.48 & 2.49 & 2.54 & 1.91 & 0.64 & 2.02 \\ 
  MC+ & 0.27 & 0.62 & 0.68 & 1.20 & 0.58 & 1.33 & 0.48 & 2.49 & 2.54 & 1.89 & 0.67 & 2.00 \\ 
 
   \hline

\end{tabular}
}
\end{center}
\end{table}

\setlength{\tabcolsep}{3pt}
\renewcommand{\arraystretch}{1.5}
\begin{table}[!ht]
\caption{Empirical prediction properties (Bias, SD, and RMSE) based on 2000 replications of Thomas processes on the domain $D$ ($\mu=1600$) for different values of $\kappa$ and for the three different scenarios. Different penalty functions are considered as well as two estimating equations, the regularized Poisson likelihood (PL) and the regularized weighted Poisson likelihood (WPL).}
\label{table:mean3}
\begin{center}
\scalebox{0.79}{
\begin{tabular}{l ccc | ccc | ccc | ccc}
  \hline
\hline
\multirow{3}{*}{Method} & \multicolumn{6}{c}{$\kappa=5 \times 10^{-4}$}  & \multicolumn{6}{c}{$\kappa=5 \times 10^{-5}$}\\
\cline{2-13}
 & \multicolumn{3}{c}{Regularized PL} & \multicolumn{3}{c}{Regularized WPL} & \multicolumn{3}{c}{Regularized PL} & \multicolumn{3}{c}{Regularized WPL} \\
\cline{2-13}
 & Bias & SD & RMSE & Bias & SD & RMSE & Bias & SD & RMSE & Bias & SD & RMSE \\ 
  \hline
\hline
& \multicolumn{12}{c}{Scenario~\ref{sce1}}\\
\hline
Oracle & 0.05 & 0.11 & 0.12 & 0.33 & 0.15 & 0.37 & 0.16 & 0.45 & 0.48 & 0.41 & 0.22 & 0.46 \\ 
  Ridge & 0.04 & 0.21 & 0.21 & 0.70 & 0.55 & 0.90 & 0.13 & 0.72 & 0.73 & 0.74 & 0.58 & 0.94 \\ 
  Lasso & 0.14 & 0.19 & 0.24 & 1.03 & 0.20 & 1.05 & 0.23 & 0.60 & 0.64 & 0.99 & 0.43 & 1.08 \\ 
  Enet & 0.11 & 0.22 & 0.24 & 1.14 & 0.29 & 1.17 & 0.20 & 0.62 & 0.65 & 1.12 & 0.43 & 1.20 \\ 
  AL & 0.04 & 0.18 & 0.18 & 0.87 & 0.18 & 0.89 & 0.16 & 0.58 & 0.60 & 0.87 & 0.42 & 0.96 \\ 
  Aenet & 0.05 & 0.18 & 0.18 & 0.96 & 0.22 & 0.99 & 0.17 & 0.58 & 0.60 & 0.90 & 0.48 & 1.02 \\ 
  SCAD & 0.19 & 0.18 & 0.26 & 1.30 & 0.34 & 1.34 & 0.14 & 0.53 & 0.55 & 1.37 & 0.51 & 1.46 \\ 
  MC+ & 0.20 & 0.18 & 0.27 & 1.33 & 0.28 & 1.36 & 0.15 & 0.53 & 0.55 & 1.38 & 0.52 & 1.48 \\ 
\hline
 & \multicolumn{12}{c}{Scenario~\ref{sce2}}\\
\hline
Oracle & 0.05 & 0.15 & 0.16 & 0.36 & 0.17 & 0.40 & 0.18 & 0.46 & 0.49 & 0.39 & 0.26 & 0.47 \\ 
  Ridge & 0.05 & 0.27 & 0.27 & 0.69 & 0.62 & 0.94 & 0.17 & 0.74 & 0.80 & 0.78 & 0.64 & 1.01 \\ 
  Lasso & 0.16 & 0.20 & 0.25 & 1.16 & 0.24 & 1.18 & 0.23 & 0.60 & 0.64 & 1.14 & 0.43 & 1.22 \\ 
  Enet & 0.17 & 0.23 & 0.29 & 1.24 & 0.24 & 1.26 & 0.23 & 0.63 & 0.67 & 1.33 & 0.42 & 1.40 \\ 
  AL & 0.07 & 0.18 & 0.20 & 0.85 & 0.18 & 0.87 & 0.18 & 0.58 & 0.61 & 0.83 & 0.41 & 0.93 \\ 
  Aenet & 0.09 & 0.19 & 0.21 & 0.94 & 0.20 & 0.96 & 0.20 & 0.59 & 0.62 & 0.92 & 0.41 & 1.01 \\ 
  SCAD & 0.26 & 0.20 & 0.33 & 1.26 & 0.51 & 1.36 & 0.19 & 0.51 & 0.55 & 1.31 & 0.60 & 1.44 \\ 
  MC+ & 0.26 & 0.20 & 0.33 & 1.31 & 0.55 & 1.42 & 0.19 & 0.51 & 0.55 & 1.32 & 0.61 & 1.46 \\ 
\hline
& \multicolumn{12}{c}{Scenario~\ref{sce3}}\\
\hline
Oracle & 0.13 & 0.31 & 0.34 & 0.43 & 0.18 & 0.47 & 0.31 & 0.96 & 1.01 & 0.75 & 0.35 & 0.83 \\ 
  Ridge & 0.11 & 0.84 & 0.86 & 0.70 & 0.96 & 1.19 & 0.23 & 2.50 & 2.51 & 1.02 & 1.43 & 1.76 \\ 
  Lasso & 0.12 & 0.64 & 0.65 & 1.14 & 0.29 & 1.17 & 0.22 & 2.41 & 2.42 & 1.40 & 0.61 & 1.52 \\ 
  Enet & 0.13 & 0.71 & 0.73 & 1.35 & 0.30 & 1.39 & 0.23 & 2.42 & 2.43 & 1.63 & 0.56 & 1.73 \\ 
  AL & 0.14 & 0.55 & 0.57 & 0.89 & 0.18 & 0.91 & 0.22 & 2.37 & 2.38 & 1.12 & 0.67 & 1.31 \\ 
  Aenet & 0.15 & 0.56 & 0.58 & 1.00 & 0.22 & 1.03 & 0.22 & 2.36 & 2.37 & 1.26 & 0.64 & 1.41 \\ 
  SCAD & 0.24 & 0.58 & 0.62 & 1.41 & 0.40 & 1.47 & 0.24 & 2.09 & 2.10 & 1.50 & 0.68 & 1.65 \\ 
  MC+ & 0.24 & 0.58 & 0.63 & 1.44 & 0.42 & 1.50 & 0.24 & 2.09 & 2.10 & 1.49 & 0.71 & 1.65 \\ 
\hline

\end{tabular}
}
\end{center}
\end{table}

Tables \ref{table:mean} and \ref{table:mean3} give the prediction properties of the estimates in terms of biases, standard deviations (SD), and square root of mean squared errors (RMSE), some criterions we define by
\begin{align*}
\mathrm{Bias}&=\left [ {\sum_{j=1}^p { \{\hat{\mathbb{E}}(\hat \beta_j)-\beta_j\}^2}}  \right ]^\frac{1}{2},
\mathrm{SD}=\left [ {\sum_{j=1}^p { \hat \sigma_j^2}}  \right ]^\frac{1}{2},
\mathrm{RMSE}=\left [ {\sum_{j=1}^p { \hat{\mathbb{E}}(\hat \beta_j-\beta_j)^2}}  \right ]^\frac{1}{2},
\end{align*}
where $\hat{\mathbb{E}}(\hat \beta_j)$ and $ \hat \sigma_j^2$ are respectively the empirical mean and variance of the estimates $\hat \beta_j$, for $j=1,\ldots,p$, where $p=20$ for scenarios~\ref{sce1} and \ref{sce2}, and $p=15$ for scenario~\ref{sce3}.

In general, the properties improve with larger value of $\kappa$ and $\mu$ due to weaker spatial dependence and larger sample size. For the oracle model where the model contains only $z_1$ and $z_2$, the WPL estimates are more efficient than the PL estimates, particularly in the more clustered case, agreeing with the findings by \cite{guan2010weighted}.

When the regularization methods are applied, the bias increases in general, especially when we consider the penalized WPL method.  The regularized WPL has a larger bias since this method does not select $z_2$ much more frequently. Furthermore, weighted method seems to introduce extra bias, even though the regularization is not considered as in the oracle model. For a low clustered process, the SD using the penalized WPL is similar to that of the penalized PL which may be because of the weaker dependence represented by larger $\kappa$, making weight surface $w(\cdot)$ closer to 1. However, a larger RMSE is obtained from the penalized WPL. When we observe the more clustered process, we obtain smaller SD using the penalized WPL which explains why in some cases (mainly scenario~\ref{sce3}) the RMSE gets smaller.

For the ridge method, the bias is closest to that of the oracle model, but it has the largest SD. Among the regularization methods, the adaptive lasso method has the best performance in terms of prediction.

Considering scenarios~\ref{sce1} and \ref{sce2}, we obtain best properties when we apply the penalized PL with adaptive lasso penalty. As the design is getting much more complex for scenario~\ref{sce3}, when we use the penalized PL with adaptive lasso, the SD is doubled and even quadrupled due to the overselection of many unimportant covariates. In particular, for the more clustered process, the better properties are even obtained by applying the regularized WPL combined with adaptive lasso. From Tables \ref{table:mean} and \ref{table:mean3}, when the focus is on prediction properties, we would recommend to apply the penalized WPL combined with adaptive lasso penalty when the observed point pattern is very clustered and when covariates have a complex stucture of covariance matrix. Otherwise, the use of the penalized PL combined with adaptive lasso penalty is more favorable. Our recommendations in terms of prediction support as what we recommend in terms of selection.

\subsection{Logistic regression}
\label{sec:logi}

Our concern here is to compare the estimates of the penalized (un)weighted logistic likelihood to that of the penalized (un)weighted Poisson likelihood with different number of dummy points. We remind that the number of dummy points comes up when we discretize the integral terms in (\ref{eq:wee}) and in (\ref{eq:logilike2}). In the following, to ease the presentation, we use the term Poisson estimates (resp. logistic estimates) for parameter estimates obtained using the regularized Poisson likelihood (resp. the regularized logistic regression likelihood).

\setlength{\tabcolsep}{3pt}
\renewcommand{\arraystretch}{1.5}
\begin{table}[!ht]
\caption{Empirical selection properties (TPR, FPR, and PPV in $\%$) based on 2000 replications of Thomas processes on the domain $D$ ($\mu=1600$) for $\kappa=5 \times 10^{-5}$, for two different scenarios, and for three different numbers of dummy points. Different estimating equations are considered, the regularized (un)weighted Poisson and (un)weighted logistic regression likelihoods, employing adaptive lasso regularization method.}
\label{table:logselect} 
\centering
\begin{threeparttable}
\scalebox{0.79}{
\begin{tabular}{rr ccc | ccc | ccc | ccc}
\hline
\hline
 \multirow{3}{*}{Method} & \multirow{3}{*}{nd} & \multicolumn{6}{c}{Scenario~\ref{sce2}} & \multicolumn{6}{c}{Scenario~\ref{sce3}}\\ 
\cline{3-14}
 &  & \multicolumn{3}{c}{Unweighted} & \multicolumn{3}{c}{Weighted} & \multicolumn{3}{c}{Unweighted} & \multicolumn{3}{c}{Weighted}\\ 
\hline
\hline
 &  & TPR & FPR & PPV & TPR & FPR & PPV & TPR & FPR & PPV & TPR & FPR & PPV \\ 
  \hline
  \multirow{3}{*}{Poisson} & 20 & 96 & 35 & 32 & 53 & 0\tnote{*} & 96 & 98 & 82 & 16 & 47 & 2 & 79 \\ 
   & 40 & 95 & 6 & 77 & 52 & 0\tnote{*} & 95 & 98 & 83 & 16 & 46 & 2 & 77  \\ 
   & 80  & 95 & 4 & 83 & 50 & 0\tnote{*} & 94 & 98 & 83 & 16 & 43 & 2 & 74  \\ 
\hline
  \multirow{3}{*}{Logistic} & 20 & 94 & 11 & 60 & 49 & 0\tnote{*} & 91 & 98 & 72 & 20 & 41 & 2 & 73 \\ 
   & 40 & 94 & 8 & 67 & 50 & 0\tnote{*} & 93 & 99 & 81 & 16 & 43 & 2 & 74  \\ 
  & 80  & 94 & 5 & 77 & 50 & 0\tnote{*} & 93 &  99 & 83 & 16 & 42 & 2 & 73 \\ 
   \hline
 \end{tabular}
}
\begin{tablenotes}
\item[*] Approximate value
\end{tablenotes}
\end{threeparttable}
\end{table}

\setlength{\tabcolsep}{3pt}
\renewcommand{\arraystretch}{1.5}
\begin{table}[ht]
\caption{Empirical prediction properties (Bias, SD, and RMSE) based on 2000 replications of Thomas processes on the domain $D$ ($\mu=1600$) for $\kappa=5 \times 10^{-5}$, for two different scenarios, and for three different numbers of dummy points. Different estimating equations are considered, the regularized (un)weighted Poisson and (un)weighted logistic regression likelihoods, employing adaptive lasso regularization method.}
\label{table:logpredict} 
\centering
\scalebox{0.79}{
\begin{tabular}{r  r  ccc | ccc | ccc | ccc}
\hline
\hline
 \multirow{3}{*}{Method} & \multirow{3}{*}{nd} & \multicolumn{6}{c}{Scenario~\ref{sce2}} & \multicolumn{6}{c}{Scenario~\ref{sce3}}\\ 
\cline{3-14}
 &  & \multicolumn{3}{c}{Unweighted} & \multicolumn{3}{c}{Weighted} & \multicolumn{3}{c}{Unweighted} & \multicolumn{3}{c}{Weighted}\\ 
\cline{3-14}
 & & Bias & SD & RMSE & Bias & SD & RMSE & Bias & SD & RMSE & Bias & SD & RMSE \\ 
  \hline
 \hline
& & \multicolumn{12}{c}{No regularization}\\
\hline
\multirow{3}{*}{Poisson}  & 20   & 0.37 & 0.64 & 0.74 & 0.29 & 0.74 & 0.79 & 0.28 & 2.15 & 2.16 & 0.42 & 2.06 & 2.11 \\ 
   & 40 & 0.14 & 0.63 & 0.65 & 0.16 & 0.73 & 0.75 & 0.33 & 2.47 & 2.50 & 0.42 & 2.32 & 2.35 \\ 
  & 80 & 0.17 & 0.64 & 0.66 & 0.11 & 0.75 & 0.76 & 0.26 & 2.57 & 2.58 & 0.43 & 2.40 & 2.43 \\ 
\hline
    \multirow{3}{*}{Logistic} & 20 & 0.03 & 0.69 & 0.69 & 0.32 & 1.34 & 1.37 & 0.20 & 2.31 & 2.32 & 0.36 & 2.95 & 2.97 \\ 
   & 40 & 0.07 & 0.60 & 0.61 & 0.12 & 0.96 & 0.97 & 0.23 & 2.31 & 2.32 & 0.37 & 2.56 & 2.58 \\ 
  & 80 & 0.10 & 0.60 & 0.61 & 0.14 & 0.81 & 0.82 & 0.25 & 2.36 & 2.38 & 0.42 & 2.38 & 2.42 \\ 
\hline
& & \multicolumn{12}{c}{Adaptive lasso}\\
\hline
 \multirow{3}{*}{Poisson}  & 20   & 0.30 & 0.59 & 0.67 & 0.86 & 0.47 & 0.98 & 0.30 & 2.00 & 2.03 & 1.14 & 0.68 & 1.33 \\ 
   & 40 & 0.20 & 0.58 & 0.61 & 0.86 & 0.49 & 0.99 & 0.33 & 2.33 & 2.35 & 1.18 & 0.70 & 1.37 \\ 
  & 80 & 0.18 & 0.59 & 0.62 & 0.88 & 0.51 & 1.02 & 0.28 & 2.41 & 2.43 & 1.22 & 0.71 & 1.41 \\ 
\hline
    \multirow{3}{*}{Logistic} & 20 & 0.19 & 0.50 & 0.53 & 0.95 & 0.55 & 1.09 & 0.23 & 2.06 & 2.07 & 1.26 & 0.73 & 1.45 \\ 
   & 40 & 0.18 & 0.52 & 0.55 & 0.89 & 0.52 & 1.03 & 0.23 & 2.15 & 2.16 & 1.22 & 0.72 & 1.42 \\ 
  & 80 & 0.18 & 0.55 & 0.58 & 0.89 & 0.52 & 1.03 & 0.25 & 2.21 & 2.22 & 1.24 & 0.71 & 1.43 \\ 
  \hline
\end{tabular}
}
\end{table}

We consider three different numbers of dummy points denoted by $\texttt{nd}^2$. By these different numbers of dummy points, we want to observe the properties with three different situations: (a) $\texttt{nd}^2 < m$, (b) $\texttt{nd}^2 \approx m$, and (c) $\texttt{nd}^2 > m$, where $m$ is the number of points. In the following, $m \approx 1600$ and $\texttt{nd}^2=$ 400, 1600, and 6400. Note that the choice by default from the Poisson likelihood in $ \texttt{spatstat}$ corresponds to case (c). \cite{baddeley2014logistic} showed that for datasets with very large number of points and for very structured point processes, the logistic likelihood method is clearly preferable as it requires a smaller number of dummy points to perform quickly and efficiently. We want to investigate a similar comparison when these methods are regularized.

We only repeat the results for $\kappa=5 \times 10^{-5}$ and $\mu=1600$, and for scenarios~\ref{sce2} and \ref{sce3}. We use the same selection and prediction indices examined in Section \ref{result} and consider only the adaptive lasso method. 

Table \ref{table:logselect} presents selection properties for the Poisson and logistic likelihoods with adaptive lasso regularization. For unweighted versions of the procedure, the regularized logistic method outperforms the regularized Poisson method when $\texttt{nd}=20$, i.e. when the number of dummy points is much smaller than the number of points. When  $\texttt{nd}^2 \approx m$ or $\texttt{nd}^2 > m$, the methods tend to have similar performances. When we consider weighted versions of the regularized logistic and Poisson likelihoods, the results do not change that much with $\texttt{nd}$ and the regularized Poisson likelihood method slightly outperforms the regularized logistic likelihood method. In addition, for scenario~\ref{sce3} which considers a more complex situation, the methods tend to select the noisy covariates much more frequently.

Empirical biases, standard deviation and square root of mean squared errors are presented in Table \ref{table:logpredict}. We include all empirical results for the standard Poisson and logistic estimates (i.e. no regularization is considered). Let us first consider the unweighted methods with no regularization. The logistic method clearly has smaller bias, especially when $\texttt{nd}=20$, which explains why in most situations the RMSE is smaller. However, for the weighted methods, although the logistic method has smaller bias in general, it produces much larger SD, leading to larger RMSE for all cases. When we compare the weighted and the unweighted methods for logistic estimates, in general, not only do we fail to reduce the SD, but we also have larger bias. When the adaptive lasso regularization is considered, combined with the unweighted methods, we can preserve the bias in general and simultaneously improve the SD, and hence improve the RMSE. The logistic likelihood method slightly outperforms the Poisson likelihood method. When the weighted methods are considered, we obtain smaller SD, but we have larger bias. For weighted versions of the Poisson and logistic likelihoods, the results do not change that much with $\texttt{nd}$ and the weighted Poisson method slightly outperforms the weighted logistic method. From Tables~\ref{table:logselect} and \ref{table:logpredict}, when the number of dummy points can be chosen as $\texttt{nd}^2 \approx m$ or $\texttt{nd}^2 > m$, we would recommend to apply the Poisson likelihood method. When the number of dummy points should be chosen as $\texttt{nd}^2 < m$, the logistic likelihood method is more favorable. Our recommendations regarding whether weighted or unweighted methods follow the ones as in Section \ref{result}.

\section{Application to forestry datasets}
\label{sec8}
In a 50-hectare region ($D=1,000\mathrm{m} \times  500\mathrm{m}$) of the tropical moist forest of Barro Colorado Island (BCI) in central Panama, censuses have been carried out where all free-standing woody stems at least 10 mm diameter at breast height were identified, tagged, and mapped, resulting in maps of over 350,000 individual trees with more than 300 species \citep[see][]{condit1998tropical, hubbell1999light, hubbell2005barro}. It is of interest to know how the very high number of different tree species continues to coexist, profiting from different habitats determined by e.g. topography or soil properties \citep[see e.g.][]{waagepetersen2007estimating, waagepetersen2009two}. In particular, the selection of covariates among topological attributes and soil minerals as well as the estimation of their coefficients are becoming our most concern.

\begin{figure}[!ht]
\begin{center}
\graphicspath{{d:/Figure/}}
\renewcommand{\arraystretch}{0}
\begin{tabular}{l l}
\includegraphics[width=0.45\textwidth]{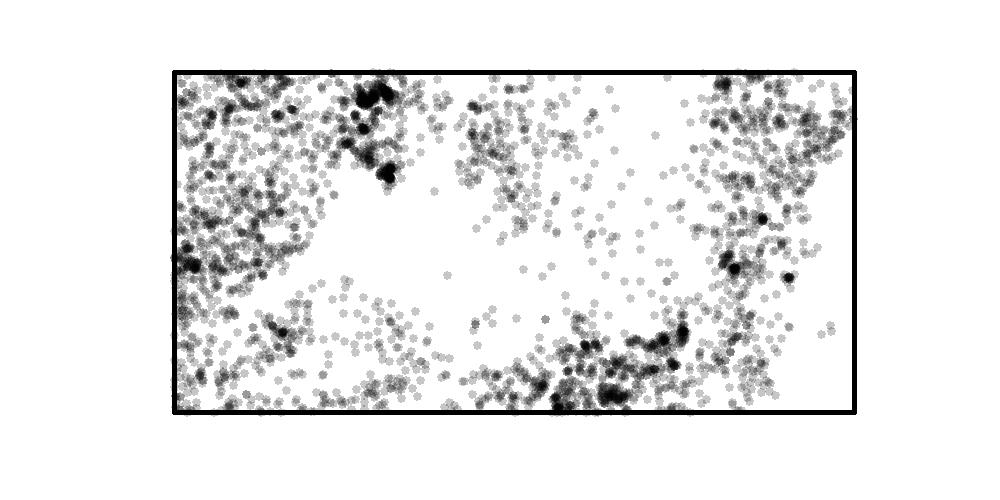} & \includegraphics[width=0.45\textwidth]{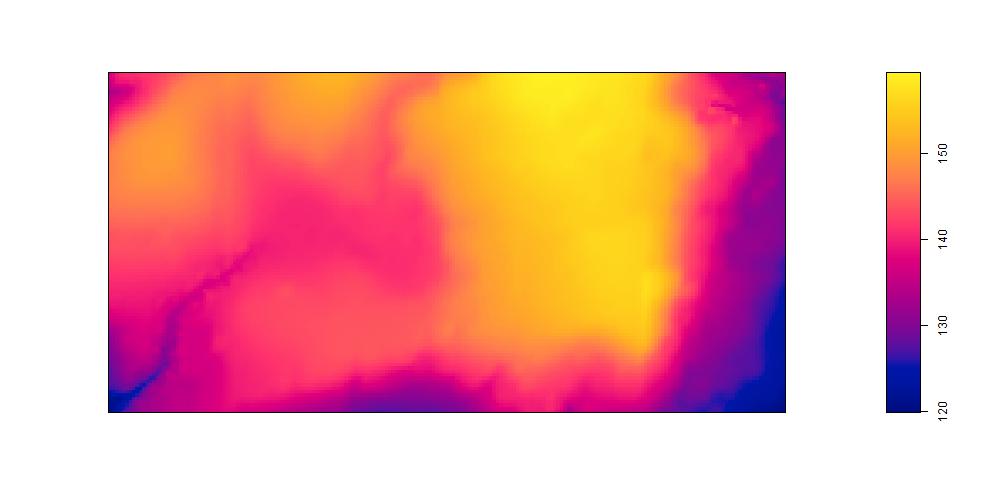} \\
\includegraphics[width=0.45\textwidth]{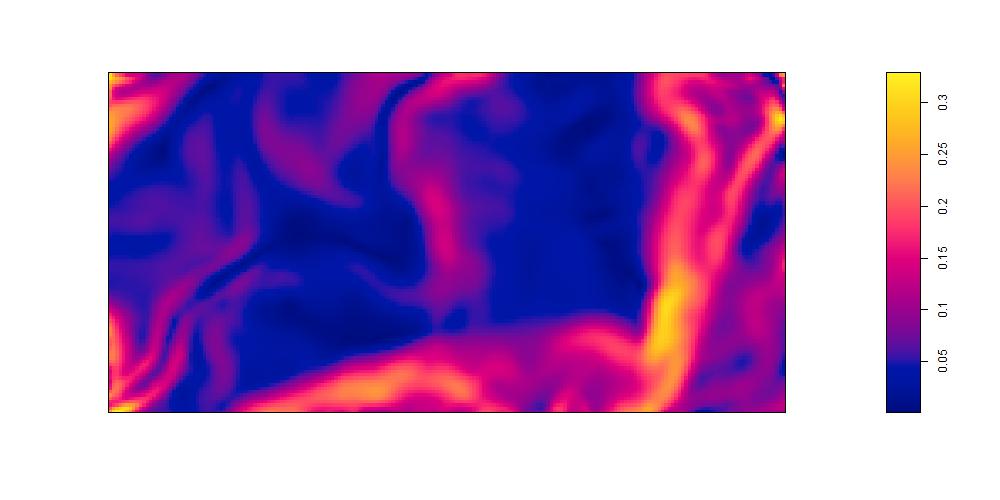} & \includegraphics[width=0.45\textwidth]{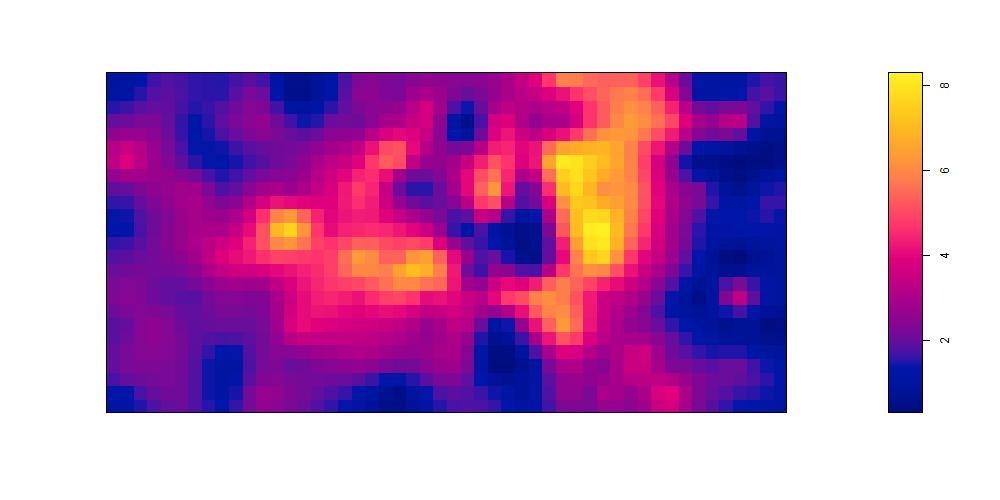}
\end{tabular}
\caption{Maps of locations of BPL trees (top left), elevation (top right), slope (bottom left), and concentration of phosporus (bottom right).}
\label{fig:app}
\end{center}
\end{figure}

We are particularly interested in analyzing the locations of 3,604 {\em Beilschmiedia pendula Lauraceae} (BPL) tree stems. We model the intensity of BPL trees as a log-linear function of two topological attributes and 13 soil properties as the covariates. Figure \ref{fig:app} contains maps of the locations of BPL trees, elevation, slope, and concentration of Phosporus. BPL trees seem to appear in greater abundance in the areas of high elevation, steep slope, and low concentration of Phosporus. The covariates maps are depicted in Figure \ref{sim3}.

\setlength{\tabcolsep}{3pt}
\renewcommand{\arraystretch}{1.5}
\begin{table}[!ht]
\caption{Barro Colorado Island data analysis: Parameter estimates of the regression coefficients for Beilschmiedia pendula Lauraceae trees applying regularized (un)weighted Poisson and logistic regression likelihoods with adaptive lasso regularization.}
\label{table:app}
\begin{center}
\scalebox{0.79}{
\begin{tabular}{  l  cc  | cc }
\hline
\hline
& \multicolumn{2}{c}{Unweighted method} & \multicolumn{2}{c}{Weighted method}\\
 \cline{2-5}
& Poisson estimates & Logistic estimates & Poisson estimates & Logistic estimates \\
  \hline
\hline
Elev  & 0.39  & 0.40  & 0.41  & 0.45 \\ 
  Slope  & 0.26  & 0.32  & 0.51  & 0.60 \\ 
  Al  & 0  & 0 & 0 & 0   \\ 
  B   & 0.30  & 0.30 & 0 & 0    \\ 
  Ca  & 0.10  & 0.15  & 0 & 0    \\ 
  Cu  & 0.10  & 0.12 & 0 & 0    \\ 
  Fe  & 0.05  & 0 & 0 & 0   \\ 
  K   & 0 & 0 & 0 & 0  \\ 
  Mg   & -0.17  & -0.18 & 0 & 0   \\ 
  Mn  & 0.12  & 0.13  & 0.23  & 0.24 \\ 
  P  & -0.60 & -0.60  & -0.50   & -0.52 \\ 
  Zn   & -0.43  & -0.46 & -0.35  & -0.37 \\ 
  N  & 0 & 0 & 0 & 0     \\ 
  N.min   & -0.12  & -0.10 & 0 & 0    \\ 
  pH  & -0.14  & -0.14 & 0 & 0    \\
\hline
Nb of cov.  & 12 & 11 & 5 & 5\\
\hline
\end{tabular}
}
\end{center}
\end{table}

We apply the regularized (un)weighted Poisson and the logistic likelihoods, combined with adaptive lasso regularization to select and estimate parameters. Since we do not deal with datasets which have very large number of points, we can set the default number of dummy points for Poisson likelihood as in the $\texttt{spatstat}$ package, i.e. the number of dummy points can be chosen to be larger than the number of points, to perform quickly and efficiently. It is worth emphasizing that we center and scale the 15 covariates to observe which one has the largest effect on the intensity. The results are presented in Table \ref{table:app}: 12 covariates for the Poisson likelihood and 11 for the logistic method are selected out of the 15 covariates using the unweighted methods while only 5 covariates (both for the Poisson and logistic methods) are selected using the weighted versions. The unweighted methods tend to overfit the model by overselecting unimportant covariates.

The weighted methods tend to keep out the uninformative covariates. Both Poisson and logistic estimates own similar selection and estimation results. First, we find some differences on estimation between the unweighted and the weighted methods, especially for slope and Manganese (Mn), for which the weighted methods have approximately two times larger estimators. Second, we may loose some nonzero covariates when we apply the weighted methods, even though it is only for the covariates which have relatively small coefficient. Boron (B) has high correlation with many of the other covariates, particularly with those which are not selected. This is possibly why Boron, which is selected and may have nonnegligible coefficient in the unweighted methods, is not chosen in the model. This may explain why the weighted methods introduce extra biases. However, since the situation appears to be quite close to the scenario~\ref{sce3} from the simulation study, the weighted methods are more favorable in terms of both selection and prediction.

In this application, we do not face any computational problem. Nevertheless, if we have to model a species of trees with much more points, the default value for $\texttt{nd}$ will lead to numerical problems. In such a case, the logistic likelihood would be a good alternative. 

These results suggest that BPL trees favor to live in areas of higher elevation and slope. This result is different from the findings by \cite{waagepetersen2007estimating} and \cite{guan2007thinned} which concluded based on standard error estimation that BPL trees do not really prefer either high or low altitudes. However, we have the same conclusion with the analysis by \cite{guan2010weighted} and \cite{thurman2015regularized} that BPL trees prefer to live on higher altitudes. Further, higher levels of Manganese (Mn) and lower levels of both Phosporus (P) and Zinc (Zn) concentrations in soil are associated with higher appearance of BPL trees.

\section{Conclusion and discussion}
\label{sec9}
We develop regularized versions of estimating equations based on Campbell theorem derived from the Poisson and the logistic likelihoods. Our procedure is able to estimate intensity function of spatial point processes, when the intensity is a function of many covariates and has a log-linear form. Furthermore, our procedure is also generally easy to implement in $\texttt{R}$ since we need  to combine $\texttt{spatstat}$ package with $\texttt{glmnet}$ and $\texttt{ncvreg}$ packages. We study the asymptotic properties of both regularized weighted Poisson and logistic estimates in terms of consistency, sparsity, and normality distribution. We find that, among the regularization methods considered in this paper, adaptive lasso, adaptive elastic net, SCAD, and MC+ are the methods that can satisfy our theorems.

We carry out some scenarios in the simulation study to observe selection and prediction properties of the estimates. We compare the penalized Poisson likelihood (PL) and the penalized weighted Poisson likelihood (WPL) with different penalty functions. From the results, when we deal with covariates having a complex covariance matrix and when the point pattern looks quite clustered, we recommend to apply the penalized WPL combined with adaptive lasso regularization. Otherwise, the regularized PL with adaptive lasso is more preferable. The further and more careful investigation to choose the tuning parameters may be needed to improve the selection properties. We note the bias increases quite significantly when the regularized WPL is applied. When the penalized WPL is considered, a two-step procedure may be needed to improve the prediction properties: (1) use the penalized WPL combined with adaptive lasso to chose the covariates, then (2) use the selected covariates to obtain the estimates. This post-selection inference procedure has not been investigated in this paper.

We also compare the estimates obtained from the Poisson and the logistic likelihoods. When the number of dummy points can be chosen to be either similar to or larger than the number of points, we recommend the use of the Poisson likelihood method. Nevertheless, when the number of dummy points should be chosen to be smaller than the number of points, the logistic method is more favorable.

A further work would consist in studying the situation when the number of the covariates is much larger than the sample size. In such a situation, the coordinate descent algorithm used in this paper may cause some numerical troubles. The Dantzig selector procedure introduced by \cite{candes2007dantzig} might be a good alternative as the implementaion for linear models (and for generalized linear models) results in a linear programming. It would be interesting to bring this approach to spatial point process setting.

\section*{Acknowledgements}
We thank A. L. Thurman who kindly shared the $\texttt{R}$ code used for the simulation study in \cite{thurman2015regularized} and P. Breheny who kindly provided his code used in $\texttt{ncvreg}$ $\texttt{R}$ package. We also thank R. Drouilhet for technical help. The research of J.-F. Coeurjolly is funded by ANR-11-LABX-0025 Persyval-lab (2011, project Oculo-Nimbus and Persyvact). The research of F. Letu\'e is funded by ANR-11-LABX-0025 Persyval-lab (project Persyvact2). The BCI soils data sets were collected and analyzed by J. Dalling, R. John, K. Harms, R. Stallard and J. Yavitt with support from NSF DEB021104,021115, 0212284,0212818 and OISE 0314581, and STRI Soils Initiative and CTFS and assistance from P. Segre and J. Trani. Datasets are available at the Center for Tropical Forest Science website \linebreak $\texttt{http://ctfs.si.edu/webatlas/datasets/bci/soilmaps/BCIsoil.html}$.

\bibliographystyle{plainnat}
\bibliography{refejs}

\appendix


\section{Auxiliary Lemma}  \label{sec:auxLemma}

The following result is used in the proof of Theorems~\ref{the1}-\ref{the2}. Throughout the proofs, the notation $\mathbf X_x = O_{\mathrm P} (x_n)$ or $\mathbf X_b = o_{\mathrm P} (x_n)$ for a random vector $\mathbf X_n$ and a sequence of real numbers $x_n$ means that $\|\mathbf X_n\|=O_{\mathrm P}(x_n)$ and $\|\mathbf X_n\|=o_{\mathrm P}(x_n)$. In the same way for a vector $\mathbf V_n$ or a squared matrix $\mathbf M_n$, the notation $\mathbf V_n=O(x_n)$ and  $\mathbf M_n=O(x_n)$ mean that $\|\mathbf V_n\|=O(x_n)$ and  $\|\mathbf M_n\|=O(x_n)$.

\begin{lemma} \label{lem:clt}
Under the conditions ($\mathcal C$.\ref{C:Dn})-($\mathcal C$.\ref{C:BnCn}), the following convergence holds in distribution as $n\to \infty$
	\begin{align}
\{ \mathbf{B}_n(w; \boldsymbol \beta_{0})+\mathbf{C}_n(w; \boldsymbol \beta_{0})\}^{-1/2}\ell^{(1)}_n(w; \boldsymbol \beta_0) \xrightarrow{d} \mathcal{N}(\mathbf{0},\mathbf{I}_{p}) \label{eq:normal}.
\end{align}
Moreover as $n\to \infty$,
\begin{equation}
|D_n|^{- \frac {1}{2}}\ell^{(1)}_n(w;\boldsymbol \beta_0)   =O_\mathrm{P}(1) \label{eq:ln}.	
\end{equation}
\end{lemma}

\begin{proof}
Let us first note that using Campbell Theorems~\eqref{eq:campbell}-\eqref{eq:campbell2}
\begin{align*}
\mathrm{Var}[ \ell^{(1)}_n(w; \boldsymbol \beta_0)]= \mathbf{B}_n(w; \boldsymbol \beta_{0})+\mathbf{C}_n(w; \boldsymbol \beta_{0}).
\end{align*}
The proof of~\eqref{eq:normal} follows~\cite{coeurjolly2014variational}. Let $C_i=i+(-1/2,1/2]^d$ be the unit box centered at $i \in \mathbb{Z}^d$ and define $\mathscr{I}_n=\{i \in \mathbb{Z}^d, C_i \cap D_n \neq \emptyset \}$. Set $D_n={\displaystyle \bigcup_{i \in \mathscr{I}_n} C_{i,n}}$, where $C_{i,n}=C_i \cap D_n$. We have
\[
\ell^{(1)}_n(w;\boldsymbol \beta_0)={\sum_{i \in \mathscr{I}_n} Y_{i,n}}	
\]
where
\[
Y_{i,n}=\!\!\!\!\!\sum_{u \in \mathbf{X} \cap C_{i,n}} \!\!\! w(u)\mathbf{z}(u) - \int_{ C_{i,n}} w(u)\mathbf{z}(u)\exp(\boldsymbol\beta_0^\top \mathbf z(u) )\mathrm{d}u.	
\]
For any $n \geq 1$ and any $i \in \mathscr{I}_n$, $Y_{i,n}$ has zero mean, and by condition ($\mathcal C$.\ref{C:rhok}),
\begin{align}
{\displaystyle \sup_{n \geq 1} \sup_{i \in \mathscr{I}_n} \mathbb{E}(\|Y_{i,n}\|^{2+\delta})} < \infty. \label{eq:12}
\end{align}

If we combine (\ref{eq:12}) with conditions  ($\mathcal C$.\ref{C:Dn})-($\mathcal C$.\ref{C:BnCn}), we can apply \citet[][Theorem 4]{karaczony2006central}, a central limit theorem for triangular arrays of random fields, to obtain~\eqref{eq:normal} which also implies that
\[
	\{ \mathbf{B}_n(w; \boldsymbol \beta_{0})+\mathbf{C}_n(w; \boldsymbol \beta_{0})\}^{-1/2}\ell^{(1)}_n(w; \boldsymbol \beta_0) =  O_{\mathrm P}(1)
\]
as $n\to \infty$. The second result~\eqref{eq:ln} is deduced from condition~($\mathcal C$.\ref{C:BnCn}) which in particular implies that $|D_n|^{1/2} \{ \mathbf{B}_n(w; \boldsymbol \beta_{0})+\mathbf{C}_n(w; \boldsymbol \beta_{0})\}^{-1/2} =  O(1)$.


\end{proof}


\section{Proof of Theorem~\ref{the1}} \label{proof1}

In the proof of this result and the following ones, the notation $\kappa$ stands for a generic constant which may vary from line to line. In particular this constant is independent of $n$, $\boldsymbol \beta_0$ and $\mathbf k$.

\begin{proof}
Let $ d_n = |D_n|^{-1/2}+a_n$, and $\mathbf{k}=\{k_1, k_2, \ldots, k_p\}^\top \in \mathbb{R}^p $. We remind the reader that the estimate of $\boldsymbol\beta_0$ is defined as the maximum of the function $Q$ (given by~\eqref{eq:qwee}) over $\Theta$, an open convex bounded set of $\mathbb R^p$. For any $\mathbf k$ such that $\|\mathbf k\|\leq K<\infty$, $\boldsymbol \beta_0 + d_n \mathbf k \in \Theta$ for $n$  sufficiently large. Assume this is valid in the following.
To prove Theorem~\ref{the1}, we follow the main argument by~\cite{fan2001variable} and aim at proving that for any given $\epsilon>0$, there exists $K>0$ such that for $n$ sufficiently large
\begin{equation}
\label{eq:15}
\mathrm{P}\bigg(\sup_{\|\mathbf{k}\| = K} \Delta_n(\mathbf k)>0\bigg)\leq \epsilon,
\quad \mbox{ where } \Delta_n(\mathbf k) = Q(w;\boldsymbol \beta_0+d_n\mathbf{k})-Q(w;\boldsymbol \beta_0).
\end{equation}
Equation~\eqref{eq:15} will imply that with probability at least $1-\epsilon$, there exists a local maximum in the ball $\{\boldsymbol \beta_0+d_n\mathbf{k}:\|\mathbf{k}\| \leq K\}$, and therefore  a local maximizer $\boldsymbol{\hat{\beta}}$ such that $\|{ \boldsymbol {\hat \beta}-\boldsymbol \beta_0}\|=O_\mathrm{P}(d_n)$.  We decompose $\Delta_n(\mathbf k)$ as $\Delta_n(\mathbf k)= T_1+T_2$ where
\begin{align*}
	T_1 = & \ell_n(w;\boldsymbol \beta_0+d_n\mathbf{k})-\ell_n(w; \boldsymbol \beta_0) \\
	T_2 = & |D_n|{\sum_{j=1}^p \big( p_{\lambda_{n,j}}(|\beta_{0j}|)}- p_{\lambda_{n,j}}(|\beta_{0j}+d_nk_j|)\big).
\end{align*}
Since $\rho(u;\cdot)$ is infinitely continuously differentiable and $\ell_n^{(2)}(w;\boldsymbol\beta) =-\mathbf A_n(w;\boldsymbol\beta)$, then using a second-order Taylor expansion there exists $t\in (0,1)$ such that
\begin{align*}
	T_1 =& \, d_n \mathbf k^\top \ell_n^{(1)}(w;\boldsymbol \beta_0) - \frac12d_n^2\mathbf k^\top \mathbf{A}_n(w;\boldsymbol \beta_0) \mathbf k  \\
	&+ \frac12d_n^2\mathbf k^\top \left( \mathbf{A}_n(w;\boldsymbol \beta_0) -\mathbf{A}_n(w;\boldsymbol \beta_0 + td_n \mathbf k) \right) \mathbf k .	
\end{align*}
Since $\Theta$ is convex and bounded and since $w(\cdot)$ and $\mathbf z(\cdot)$ are uniformly bounded by conditions~($\mathcal C$.\ref{C:Theta})-($\mathcal C$.\ref{C:cov}), there exists a nonnegative constant $\kappa$ such that
\[
	\frac12\|\mathbf{A}_n(w;\boldsymbol \beta_0) -\mathbf{A}_n(w;\boldsymbol \beta_0 + td_n \mathbf k) \| \leq  \kappa d_n |D_n|.
\]
Let $\nu_{\min}(\mathbf M)$ be the smallest eigenvalue of a squared matrix $\mathbf M$. By condition ($\mathcal C$.\ref{C:An}), 
\[
	\check \nu := \liminf_{n\to \infty} \nu_{\min}(|D_n|^{-1}\mathbf{A}_n(w;\boldsymbol \beta_0)) = \liminf_{n\to \infty} \frac{\mathbf k^\top \left( |D_n|^{-1}\mathbf A_n(w;\boldsymbol\beta_0)\right) \mathbf k}{\|\mathbf k\|^2} >0.
\]
Hence
\[
	T_1 \leq d_n \|\ell_n^{(1)}(w;\boldsymbol \beta_0)\| \, \| \mathbf k \|  - \frac{\check \nu}2 d_n^2 |D_n| \|\mathbf k\|^2 +\kappa d_n^3 |D_n|.
\]

Regarding the term $T_2$,
\[
T_2\leq T_2^\prime := |D_n|{\sum_{j=1}^s \big( p_{\lambda_{n,j}}(|\beta_{0j}|)}- p_{\lambda_{n,j}}(|\beta_{0j}+d_nk_j|)\big) 	
\]
since for any $j$ the penalty function $p_{\lambda_{n,j}}$ is nonnegative and $p_{\lambda_{n,j}}(|\beta_{0j}|)=0$ for $j=s+1,\dots,p$.

Since $d_n|D_n|^{1/2}= O(1)$, then by ($\mathcal C$.\ref{C:plambda}), for $n$ sufficiently large, $p_{\lambda_{n,j}}$ is twice continuously differentiable for every $\beta_j = \beta_{0j}+t d_n k_j$ with $t\in (0,1)$. Therefore using a third-order Taylor expansion, there exist $t_j \in (0,1)$, $j=1,\dots,s$ such that
\begin{align*}
-T_2^\prime &=d_n|D_n|\sum_{j=1}^s k_j p_{\lambda_{n,j}}^\prime(|\beta_{0j}|) \sign(\beta_{0,j})  + \frac12 d_n^2 |D_n|\sum_{j=1}^s k_j^2 p^{\prime\prime}_{\lambda_{n,j}}  (|\beta_{0j}|) \\
&\; +\frac16 d_n^3|D_n| \sum_{j=1}^s k_j^3 p^{\prime\prime\prime}_{\lambda_{n,j}}  (|\beta_{0j}+t_j d_n k_j|).
\end{align*}
Now by definition of $a_n$ and $c_n$ and from condition~($\mathcal C$.\ref{C:plambda}), we deduce that there exists $\kappa$ such that
\begin{align*}
	T_2^\prime &\leq a_n d_n |D_n| \, |\mathbf k^\top \mathbf 1| + \frac{1}{2} c_n d_n^2|D_n| \|\mathbf k\|^2	+ \kappa d_n^3 |D_n|\\
	&\leq \sqrt s a_n d_n |D_n| \|\mathbf k \| + \frac{1}{2} c_n d_n^2|D_n| \|\mathbf k\|^2	+ \kappa d_n^3 |D_n|
\end{align*}
from Cauchy-Schwarz inequality. Since $c_n=o(1)$, $d_n=o(1)$ and $a_n d_n |D_n|=O(d_n^2|D_n|)$, then for $n$ sufficiently large
\[
	\Delta_n(\mathbf k) \leq d_n \|\ell_n^{(1)}(w;\boldsymbol \beta_0)\| \, \| \mathbf k \| -\frac{\check \nu}4 d_n^2 |D_n| \|\mathbf k\|^2 + 2 \sqrt s d_n^2 |D_n| \|\mathbf k\|
\]
We now return to (\ref{eq:15}): for $n$ sufficiently large
\[
\mathrm{P}\bigg({\sup_{\|\mathbf{k}\|= K}  \Delta_n(\mathbf{k})>0}\bigg) \leq 	\mathrm P \bigg(
\| \ell_n^{(1)}(w;\boldsymbol\beta_0)\| > \frac{\check \nu}4 d_n|D_n| K - 2\sqrt s d_n |D_n|
 \bigg)
\]

Since $d_n |D_n|=O(|D_n|^{1/2})$, by choosing $K$ large enough, there exists $\kappa$ such that for $n$  sufficiently large
\[
	\mathrm P \bigg( \sup_{\|\mathbf k\|=K}  \Delta_n(\mathbf k) >0\bigg) \leq \mathrm P \bigg( \|\ell_n^{(1)}(w;\boldsymbol \beta_0)\| >\kappa|D_n|^{1/2}\bigg) \leq \epsilon
\]
for any given $\epsilon>0$ from \eqref{eq:ln}.

\end{proof}


\section{Proof of Theorem~\ref{the2}} \label{proof2}
To prove Theorem~\ref{the2}(i), we provide Lemma~\ref{lemma1} as follows.
\begin{lemma}
\label{lemma1}
Assume the conditions ($\mathcal C$.\ref{C:Dn})-($\mathcal C$.\ref{C:BnCn}) and condition ($\mathcal C$.\ref{C:plambda}) hold. If $a_n=O(|D_n|^{-1/2})$ and $b_n|D_n|^{1/2}\to \infty$ as $n\to\infty$, then with probability tending to $1$, for any {$\boldsymbol \beta_1$} satisfying $\|{\boldsymbol \beta_1 - \boldsymbol \beta_{01}}\|=O_\mathrm{P}(|D_n|^{-1/2})$, and for any constant $K_1 > 0$,
\begin{align*}
Q\Big(w;({\boldsymbol \beta_1}^\top,\mathbf{0}^\top)^\top \Big)
= \max_{\| \boldsymbol \beta_2\| \leq K_1 |D_n|^{-1/2}}
Q\Big(w;({\boldsymbol \beta_1}^\top,{\boldsymbol \beta_2}^\top)^\top \Big).
\end{align*}
\end{lemma}
\begin{proof}
It is sufficient to show that with probability tending to $1$ as ${n\to \infty}$, for any ${\boldsymbol \beta_1}$ satisfying $\|{\boldsymbol \beta_1 -\boldsymbol \beta_{01}}\|=O_\mathrm{P}(|D_n|^{-1/2})$, for some small $\varepsilon_n=K_1|D_n|^{-1/2}$, and for $j=s+1, \ldots, p$,

\begin{equation}
\label{eq:lem1}
\frac {\partial Q(w;\boldsymbol \beta)}{\partial\beta_j}<0 \quad
\mbox { for } 0<\beta_j<\varepsilon_n, \mbox{ and}
\end{equation}

\begin{equation}
\label{eq:lem1b}
\frac {\partial Q(w;\bf \boldsymbol \beta)}{\partial\beta_j}>0 \quad
\mbox { for } -\varepsilon_n<\beta_j<0.
\end{equation}

First note that by (\ref{eq:ln}), we obtain $\| \ell^{(1)}_n(w; \boldsymbol \beta_0)\|=O_\mathrm{P}(|D_n|^{1/2})$. Second, by conditions  ($\mathcal C$.\ref{C:Theta})-($\mathcal C$.\ref{C:cov}), there exists $t\in (0,1)$ such that 
\begin{align*}
\frac {\partial \ell_n(w;\boldsymbol \beta)}{\partial\beta_j}&=\frac {\partial \ell_n{(w;\boldsymbol \beta_0)}}{\partial\beta_j}+ t {\sum_{l=1}^p \frac {\partial^2 \ell_n{(w;\boldsymbol \beta_0 + t (\boldsymbol \beta -\boldsymbol\beta_0))}}{\partial\beta_j \partial\beta_l}}(\beta_l-\beta_{0l}) \\
&=O_\mathrm{P}(|D_n|^{1/2})+O_\mathrm{P}(|D_n||D_n|^{-1/2})=O_\mathrm{P}(|D_n|^{1/2}).
\end{align*}
Third, let $0<\beta_j<\varepsilon_n$ and $b_n$ the sequence given by~(\ref{eq:bn}). By condition ($\mathcal C$.\ref{C:plambda}), $b_n$ is well-defined and since by assumption $b_n|D_n|^{1/2}\to \infty$, in particular,  $b_n>0$ for $n$ sufficiently large. Therefore, for $n$ sufficiently large,
\begin{align*}
\mathrm{P} \left ( \frac {\partial Q(w;\boldsymbol \beta)}{\partial\beta_j}<0 \right)&=\mathrm{P} \left ( \frac {\partial \ell_n(w;\boldsymbol \beta)}{\partial\beta_j} - |D_n|p'_{\lambda_{n,j}}(|\beta_j|)\sign(\beta_j)<0 \right)\\
&=\mathrm{P} \left ( \frac {\partial \ell_n(w;\boldsymbol \beta)}{\partial\beta_j}< |D_n|p'_{\lambda_{n,j}}(|\beta_j|) \right)\\
& \geq \mathrm{P} \left ( \frac {\partial \ell_n(w;\boldsymbol \beta)}{\partial\beta_j}< |D_n|b_n \right)\\
&= \mathrm{P} \left ( \frac {\partial \ell_n(w;\boldsymbol \beta)}{\partial\beta_j}< |D_n|^{1/2}|D_n|^{1/2}b_n \right).
\end{align*}
$\mathrm{P} \left ( {\partial Q(w;\boldsymbol \beta)}/{\partial\beta_j}<0 \right) \xrightarrow{}1 \mbox{ as } n \to \infty$
since $ {\partial \ell_n(w;\boldsymbol \beta)}/{\partial\beta_j}=O_\mathrm{P}(|D_n|^{1/2})$ and $b_n|D_n|^{1/2} \xrightarrow{} \infty$. This proves (\ref{eq:lem1}). We proceed similarly to prove (\ref{eq:lem1b}).
\end{proof}


\bigskip

\begin{proof} We now focus on the proof of Theorem~\ref{the2}. Since Theorem~\ref{the2}(i) is proved by Lemma~\ref{lemma1}, we only need to prove Theorem~\ref{the2}(ii), which is the asymptotic normality of $\boldsymbol {\hat{\beta}}_1$. As shown in Theorem~\ref{the1}, there is a root-$|D_n|$ consistent local maximizer $\boldsymbol{\hat{\beta}}$ of $Q(w;\boldsymbol \beta)$, and it can be shown that there exists an estimator $\boldsymbol {\hat{\beta}}_1$ in Theorem~\ref{the1} that is a root-$(|D_n|)$ consistent local maximizer of $ Q\Big(w;({\boldsymbol \beta_1}^\top,\mathbf{0}^\top)^\top \Big)$, which is regarded as a function of  $\boldsymbol {\beta}_1$, and that satisfies
\begin{align*}
\frac {\partial Q(w;\boldsymbol {\hat \beta})}{\partial\beta_j}=0 \quad
\mbox { for } j=1,\ldots,s
\mbox {, and } \boldsymbol{\hat \beta}=( \boldsymbol {\hat{\beta}}_1^\top,\mathbf{0}^ \top)^\top.
\end{align*}
There exists $t\in (0,1)$ and $\boldsymbol{\breve{\beta}}= \boldsymbol{\hat \beta} + t(\boldsymbol\beta_0-\boldsymbol{\hat \beta})$  such that
\begin{align}
0
=&\frac {\partial \ell_n{(w;\boldsymbol{\hat \beta})}}{\partial\beta_j}-|D_n|p'_{\lambda_{n,j}}(|\hat \beta_{j}|)\sign(\hat \beta_j) \nonumber\\
=&\frac {\partial \ell_n{(w;\boldsymbol \beta_0)}}{\partial\beta_j}+{\sum_{l=1}^s \frac {\partial^2 \ell_n{(w; \boldsymbol{\breve{\beta}})}}{\partial\beta_j \partial\beta_l}}({\hat \beta_l}-\beta_{0l})-|D_n|p'_{\lambda_{n,j}}(|\hat \beta_{j}|)\sign(\hat \beta_j) \nonumber\\
=&\frac {\partial \ell_n{(w;\boldsymbol \beta_0)}}{\partial\beta_j}+{\sum_{l=1}^s \frac {\partial^2 \ell_n{(w; \boldsymbol \beta_0)}}{\partial\beta_j \partial\beta_l}}({\hat \beta_l}-\beta_{0l})+{\sum_{l=1}^s \Psi_{n,jl}({\hat \beta_l}-\beta_{0l})} \nonumber \\
&-|D_n|p'_{\lambda_{n,j}}(|\beta_{0j}|)\sign(\beta_{0j})-|D_n|\phi_{n,j}, \label{eq:0equal}
\end{align}
where 
\begin{align*}
\Psi_{n,jl}=\frac {\partial^2 \ell_n{(w;\boldsymbol{\breve{\beta}})}}{\partial\beta_j \partial\beta_l}-\frac {\partial^2 \ell_n{(w;\boldsymbol \beta_0)}}{\partial\beta_j \partial\beta_l}
\end{align*}
and $\phi_{n,j}=p'_{\lambda_{n,j}}(|\hat \beta_{j}|)\sign(\hat \beta_j)-p'_{\lambda_{n,j}}(|\beta_{0j}|)\sign(\beta_{0j})$. We decompose $\phi_{n,j}$ as $\phi_{n,j}=T_1+T_2$ where
\[
	T_1=\phi_{n,j} \mathbb{I}(|\hat \beta_j- \beta_{0j}| \leq \tilde r_{n,j}) 
	\quad \mbox{ and }\quad 
	T_2=\phi_{n,j} \mathbb{I}(|\hat \beta_j-\beta_{0j}| > \tilde r_{n,j})
\]
and where $\tilde r_{n,j}$ is the sequence defined in the condition ($\mathcal C$.\ref{C:plambda}). Under this condition, the following Taylor expansion can be derived for the term $T_1$: there exists $t\in (0,1)$ and $\check\beta_j= \hat\beta_j+t(\beta_{0j}-\hat\beta_j)$ such that 
\begin{align*}
 T_1&  =p''_{\lambda_{n,j}}(|\beta_{0j}|)(\hat \beta_j- \beta_{0j}) \mathbb{I}(|\hat \beta_j- \beta_{0j}| \leq \tilde r_{n,j}) \\
 &\quad + \frac12(\hat\beta_{j} - \beta_{0j})^2 p'''_{\lambda_{n,j}}(|\breve{\beta}_{j}|) \mathrm{sign}(\check \beta_j) \mathbb{I}(|\hat \beta_j- \beta_{0j}| \leq \tilde r_{n,j}) \\
 &=p''_{\lambda_{n,j}}(|\beta_{0j}|)(\hat \beta_j- \beta_{0j}) \mathbb{I}(|\hat \beta_j- \beta_{0j}| \leq \tilde r_{n,j}) + O_{\mathrm P}(|D_n|^{-1})
\end{align*}
where the latter equation ensues from Theorem~\ref{the1} and condition~($\mathcal C$.\ref{C:plambda}). Again,
from Theorem~\ref{the1},  $\mathbb{I}(|\hat \beta_j- \beta_{0j}| \leq \tilde r_{n,j}) \xrightarrow{L^1} 1$  which implies that $\mathbb{I}(|\hat \beta_j- \beta_{0j}| \leq \tilde r_{n,j}) \xrightarrow{\mathrm{P}} 1$,  so $T_1=p''_{\lambda_{n,j}}(|{\beta}_{0j}|)(\hat \beta_j- \beta_{0j}) \big(1+o_{\mathrm{P}}(1)\big)+ O_{\mathrm P}(|D_n|^{-1})$.

Regarding the term $T_2$, since $p'_{\lambda}$ is a Lipschitz function, there exists $\kappa \geq 0$ such that
\[
T_2 \leq \kappa |\hat \beta_j-\beta_{0j}| \; \mathbb{I}(|\hat \beta_j-\beta_{0j}| > \tilde r_{n,j}).	
\]
By Theorem~\ref{the1}, $|\hat \beta_j- \beta_{0j}|=O_{\mathrm{P}}(|D_n|^{-1/2})$ and $\mathbb{I}(|\hat \beta_j-\beta_{0j}| > \tilde r_{n,j}) =o_{\mathrm{P}}(1)$, so $T_2 = o_{\mathrm{P}}(|D_n|^{-1/2})$ and we deduce that 
\begin{equation} \label{eq:phinj}
	\phi_{n,j} = p''_{\lambda_{n,j}}(|{\beta}_{0j}|)(\hat \beta_j- \beta_{0j}) \big(1+o_{\mathrm{P}}(1)\big)+ o_{\mathrm P}(|D_n|^{-1/2}).
\end{equation}

Let $\ell^{(1)}_{n,1}(w;\boldsymbol \beta_{0})$ (resp. $\ell^{(2)}_{n,1}(w;\boldsymbol \beta_{0})$) be the first $s$ components (resp. $s \times s$ top-left corner) of $\ell^{(1)}_{n}(w;\boldsymbol \beta_{0})$ (resp. $\ell^{(2)}_{n}(w;\boldsymbol \beta_{0})$). Let also $\boldsymbol \Psi_n$ be the $s \times s$ matrix containing $\Psi_{n,jl}, j,l=1,\ldots,s$. Finally, let the vector $\mathbf{p}'_n$, the vector $\boldsymbol \phi_n$ and the $s \times s$ matrix $\mathbf{M}_n$ be
\begin{align*}
\mathbf{p}'_n&=\{p'_{\lambda_{n,1}}(|\beta_{01}|)\sign(\beta_{01}),\ldots,p'_{\lambda_{n,s}}(|\beta_{0s}|)\sign(\beta_{0s})\}^\top, \\
\boldsymbol \phi_n&=\{\phi_{n,1},\ldots,\phi_{n,s}\}^\top, \mbox{ and}\\
\mathbf{M}_n&=\{ \mathbf{B}_{n,11}(w; \boldsymbol \beta_{0})+\mathbf{C}_{n,11}(w; \boldsymbol \beta_{0})\}^{-1/2}.
\end{align*}
We rewrite both sides of~\eqref{eq:0equal} as
\begin{equation}
\ell^{(1)}_{n,1}(w;\boldsymbol \beta_{0})+\ell^{(2)}_{n,1}(w;\boldsymbol \beta_{0})(\boldsymbol{\hat \beta}_1-\boldsymbol \beta_{01})+ \boldsymbol \Psi_n (\boldsymbol{\hat \beta}_1-\boldsymbol \beta_{01}) -|D_n| \mathbf{p}'_n-|D_n| \boldsymbol \phi_n   =0.\label{eq:0vec}
\end{equation}
By definition of $\boldsymbol \Pi_n$ given by (\ref{eq:pi}) and from~\eqref{eq:phinj}, we obtain $\boldsymbol \phi_n=\boldsymbol \Pi_n (\boldsymbol{\hat \beta}_1-\boldsymbol \beta_{01})\big(1+o_{\mathrm{P}}(1)\big)+o_{\mathrm P}(|D_n|^{-1/2})$. Using this, we deduce, by premultiplying both sides of~\eqref{eq:0vec} by $\mathbf M_n$, that
\begin{align*}
\mathbf{M}_n \ell^{(1)}_{n,1}(w;\boldsymbol \beta_{0})-&\mathbf{M}_n \big(\mathbf{A}_{n,11}(w;\boldsymbol \beta_{0})+ |D_n| \boldsymbol \Pi_n\big)(\boldsymbol{\hat \beta}_1-\boldsymbol \beta_{01}) \\
& =O(|D_n| \, \|\mathbf{M}_n \mathbf{p}'_n \|) + o_{\mathrm{P}}(|D_n| \, \|\mathbf{M}_n\boldsymbol \Pi_n (\boldsymbol{\hat \beta}_1-\boldsymbol \beta_{01}) \|)  \\
& \quad+ o_\mathrm{P} (\|\mathbf M_n \| \; |D_n|^{1/2})
+ O_{\mathrm{P}} (\|\mathbf M_n \boldsymbol \Psi_n (\boldsymbol{\hat \beta}_1-\boldsymbol \beta_{01}) \|) .
\end{align*}
 The condition ($\mathcal C$.\ref{C:BnCn}) implies that there exists an $s \times s$ positive definite matrix $\mathbf{I}_0''$ such that for all sufficiently large $n$, we have $|D_n|^{-1} (\mathbf{B}_{n,11}(w;\boldsymbol \beta_0)+\mathbf C_{n,11}(w;\boldsymbol\beta_0) )\geq \mathbf{I}_0''$, hence $\|\mathbf M_n\|=O(|D_n|^{-1/2})$.

Now, $\| \boldsymbol \Psi_n\|=O_{\mathrm{P}}(|D_n|^{1/2})$  by conditions ($\mathcal C$.\ref{C:Theta})-($\mathcal C$.\ref{C:cov}) and by Theorem~\ref{the1}, and $\|\boldsymbol{\hat \beta}_1-\boldsymbol \beta_{01}\|=O_{\mathrm{P}}(|D_n|^{-1/2})$  by Theorem~\ref{the1} and by Theorem~\ref{the2}(i). Finally, since by assumption $a_n=o(|D_n|^{-1/2})$, we deduce that
\begin{align*}
\|\mathbf{M}_n \boldsymbol \Psi_n (\boldsymbol{\hat \beta}_1-\boldsymbol \beta_{01})\|&=O_{\mathrm{P}}(|D_n|^{-1/2})=o_{\mathrm{P}}(1), \\
|D_n| \, \|\mathbf{M}_n\boldsymbol \Pi_n (\boldsymbol{\hat \beta}_1-\boldsymbol \beta_{01}) \| &= o_\mathrm{P}(1),\\
\|\mathbf M_n \| \; |D_n|^{1/2} &= O(1),\\
|D_n| \,\|\mathbf{M}_n \mathbf{p}'_n\|&=O(a_n|D_n|^{1/2})=o(1). \\
\end{align*}
Therefore, we have that
\begin{align*}
\mathbf{M}_n \ell^{(1)}_{n,1}(w;\boldsymbol \beta_{0})-\mathbf{M}_n \big(\mathbf{A}_{n,11}(w;\boldsymbol \beta_{0})+ |D_n| \Pi_n\big)(\boldsymbol{\hat \beta}_1-\boldsymbol \beta_{01}) =o_{\mathrm{P}}(1).
\end{align*}
From (\ref{eq:normal}), Theorem~\ref{the2}(i) and by Slutsky's Theorem, we deduce that
\begin{align*}
\{ \mathbf{B}_{n,11}(w; \boldsymbol \beta_{0})+\mathbf{C}_{n,11}(w; \boldsymbol \beta_{0})\}^{-1/2}
\{\mathbf{A}_{n,11}(w;\boldsymbol \beta_{0})+|D_n| \boldsymbol \Pi_n\}(\boldsymbol{\hat \beta}_1-\boldsymbol \beta_{01})&\xrightarrow{d} \mathcal{N}(0,\mathbf{I}_{s})
\end{align*}
as $n \to \infty$, which can be rewritten, in particular under ($\mathcal C$.\ref{C:An}), as 
\[
|D_n|^{1/2}\boldsymbol \Sigma_n(w;\boldsymbol \beta_{0})^{-1/2}(\boldsymbol{\hat \beta}_1-\boldsymbol \beta_{01})\xrightarrow{d}\mathcal{N}(0,\mathbf{I}_{s})	
\]
where $\mathbf \Sigma_n(w,\boldsymbol \beta_{0})$ is given by~\eqref{eq:Sigman}.
\end{proof}

\section{Map of covariates} \label{mapcov}
\begin{figure}[!ht]
\begin{center}
\graphicspath{{d:/Figure/}}
\renewcommand{\arraystretch}{0}
\begin{tabular}{l l l l l}
\includegraphics[width=0.18\textwidth]{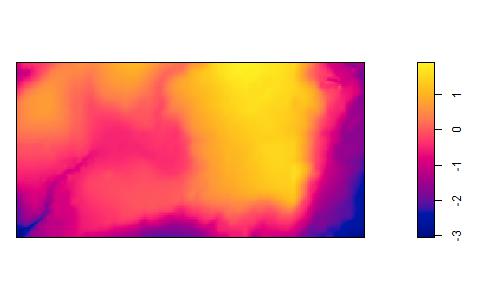} & \includegraphics[width=0.18\textwidth]{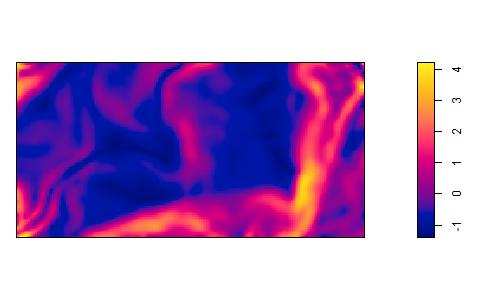} &  \includegraphics[width=0.18\textwidth]{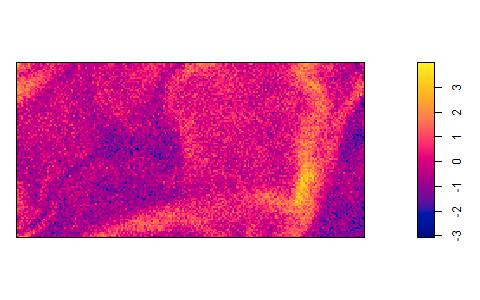} &  \includegraphics[width=0.18\textwidth]{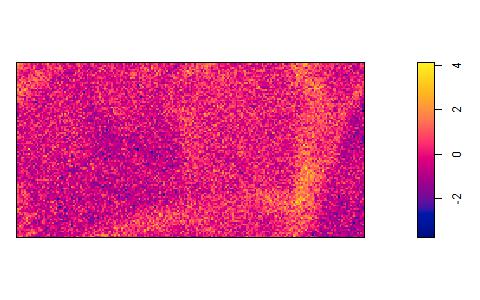} &  \includegraphics[width=0.18\textwidth]{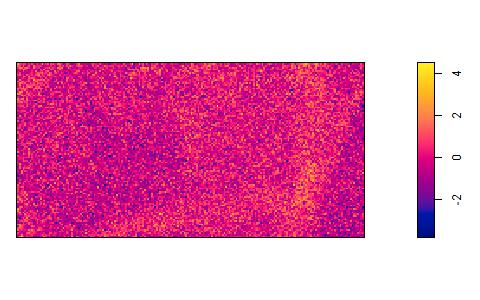}\\
\includegraphics[width=0.18\textwidth]{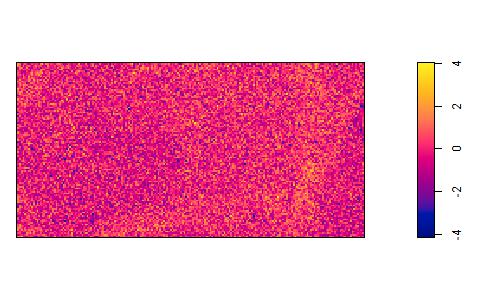} & \includegraphics[width=0.18\textwidth]{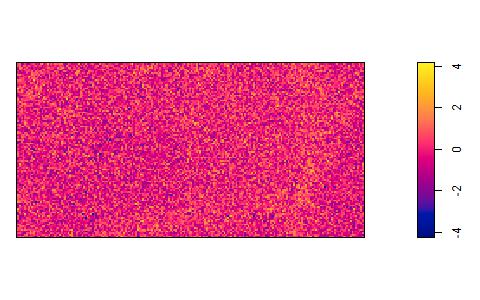} &  \includegraphics[width=0.18\textwidth]{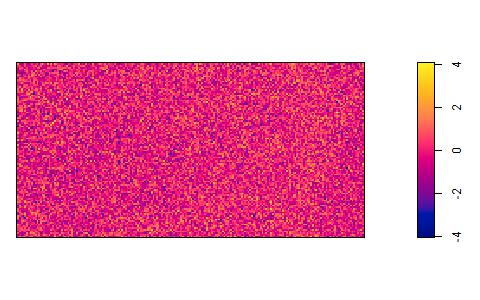} &  \includegraphics[width=0.18\textwidth]{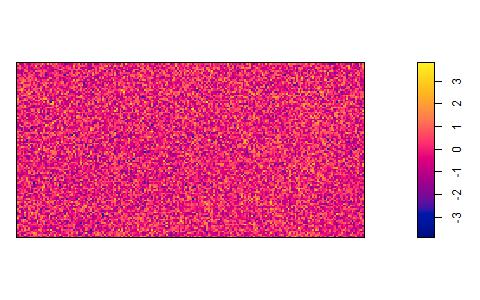} &  \includegraphics[width=0.18\textwidth]{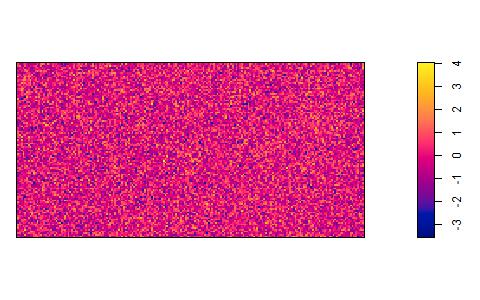}\\
\includegraphics[width=0.18\textwidth]{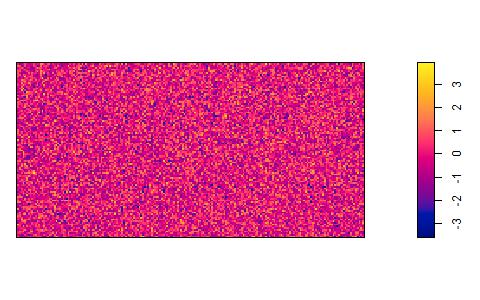} & \includegraphics[width=0.18\textwidth]{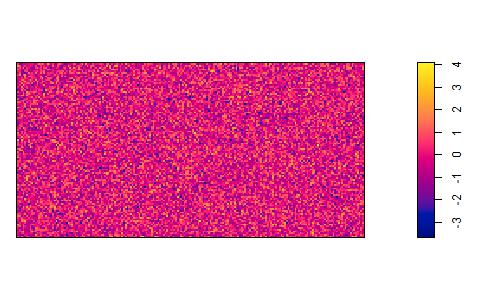} &  \includegraphics[width=0.18\textwidth]{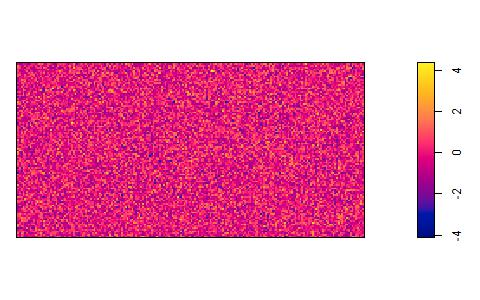} &  \includegraphics[width=0.18\textwidth]{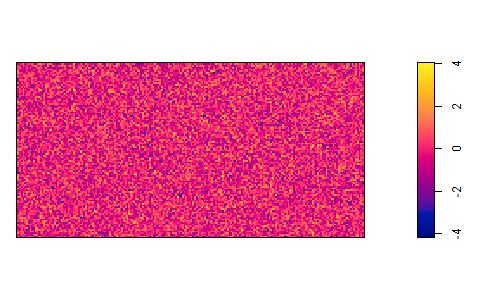} &  \includegraphics[width=0.18\textwidth]{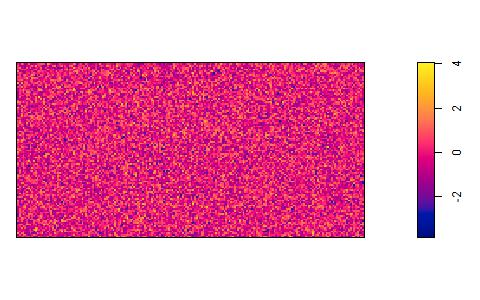}\\
\includegraphics[width=0.18\textwidth]{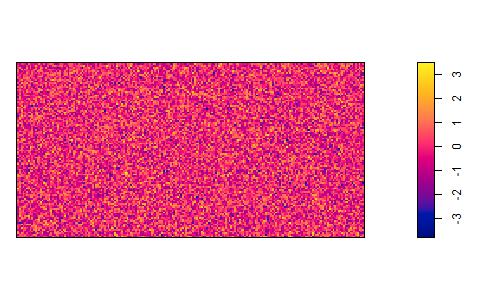} & \includegraphics[width=0.18\textwidth]{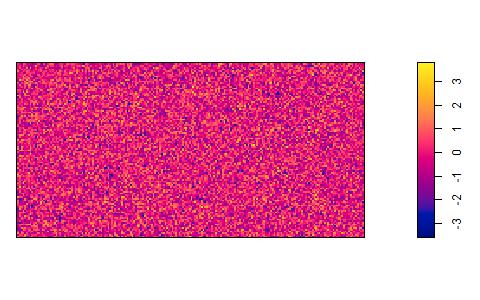} &  \includegraphics[width=0.18\textwidth]{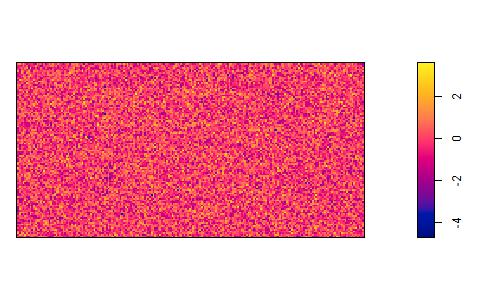} &  \includegraphics[width=0.18\textwidth]{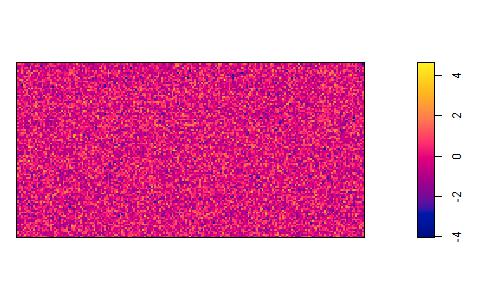} &  \includegraphics[width=0.18\textwidth]{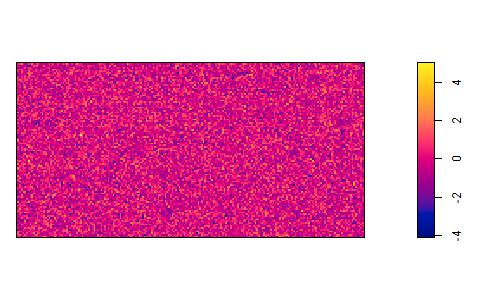}\\
\end{tabular}
\caption{Maps of covariates designed in scenario~\ref{sce2}. The first two  top left images are the elevation and the slope. The other 18 covariates are generated as standard Gaussian white noise but transformed to get multicollinearity.}
\label{sim2}
\end{center}
\end{figure}

\begin{figure}[ht]
\begin{center}
\graphicspath{{d:/Figure/}}
\renewcommand{\arraystretch}{0}
\begin{tabular}{l l l l l}
\includegraphics[width=0.18\textwidth]{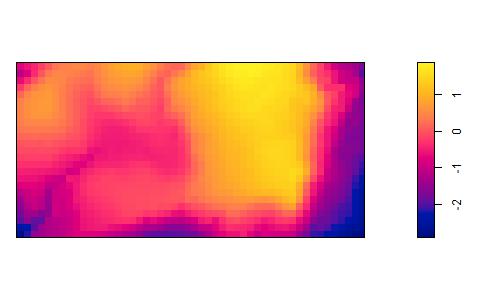} & \includegraphics[width=0.18\textwidth]{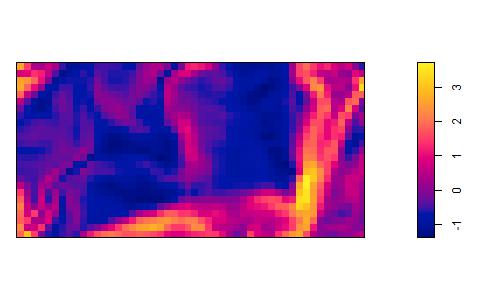} &  \includegraphics[width=0.18\textwidth]{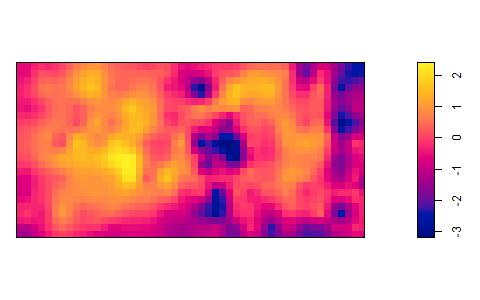} &  \includegraphics[width=0.18\textwidth]{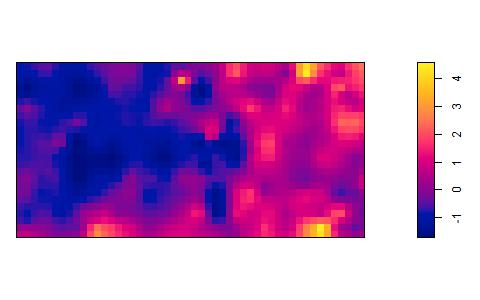} &  \includegraphics[width=0.18\textwidth]{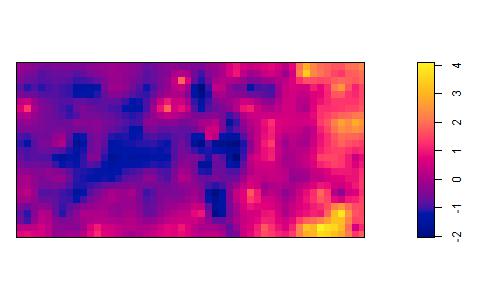}\\
\includegraphics[width=0.18\textwidth]{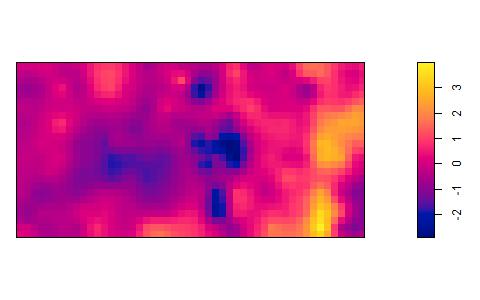} & \includegraphics[width=0.18\textwidth]{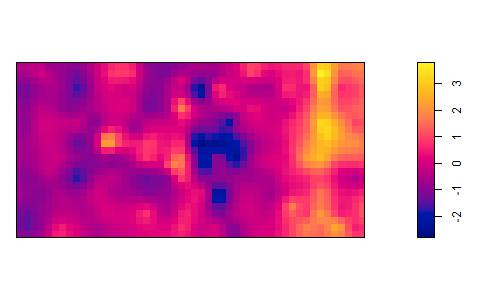} &  \includegraphics[width=0.18\textwidth]{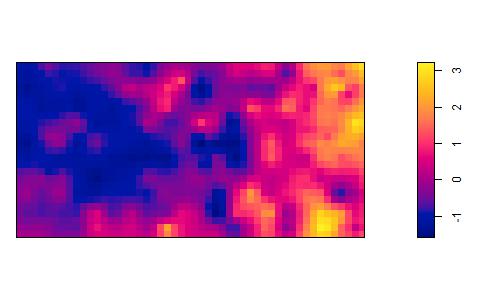} &  \includegraphics[width=0.18\textwidth]{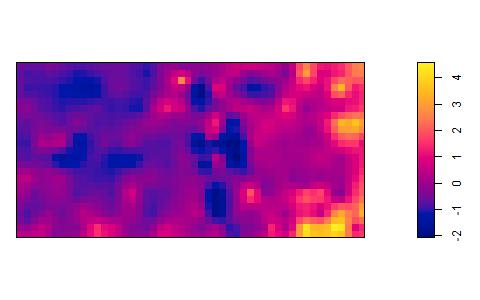} &  \includegraphics[width=0.18\textwidth]{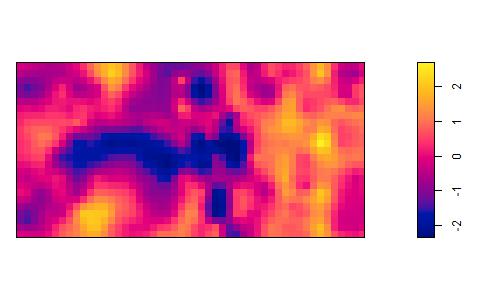}\\
\includegraphics[width=0.18\textwidth]{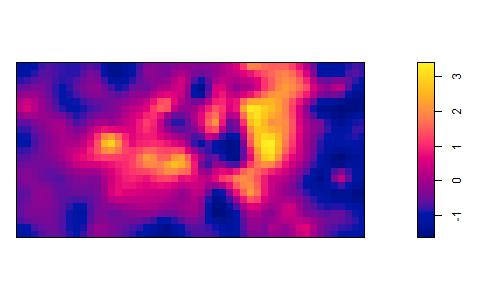} & \includegraphics[width=0.18\textwidth]{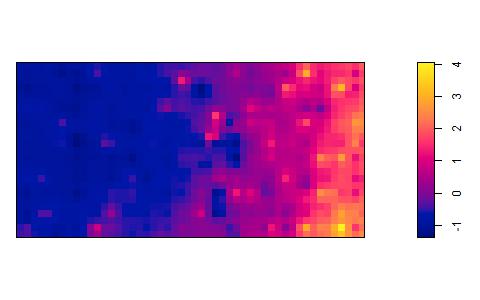} &  \includegraphics[width=0.18\textwidth]{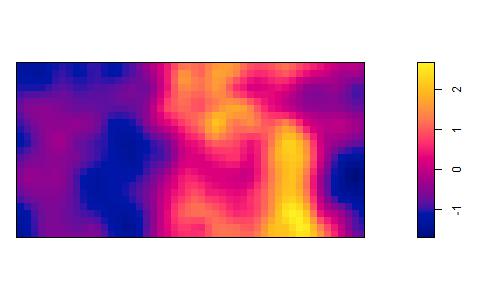} &  \includegraphics[width=0.18\textwidth]{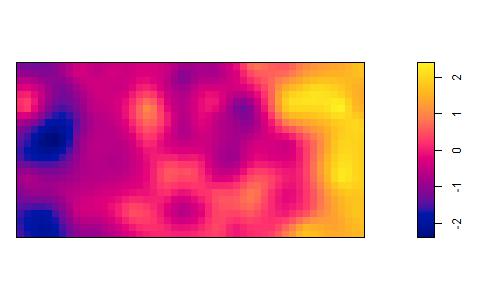} &  \includegraphics[width=0.18\textwidth]{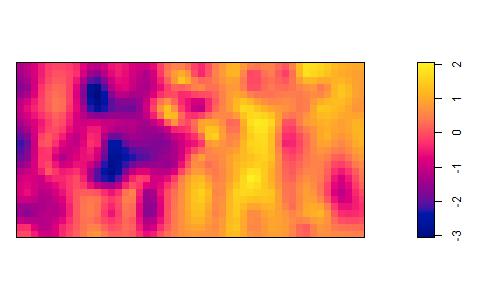}\\
\end{tabular}
\caption{Maps of covariates used in scenario~\ref{sce3} and in application. From left to right: Elevation, slope, Aluminium, Boron, and Calcium (1st row), Copper, Iron, Potassium, Magnesium, and Manganese (2nd row), Phosporus, Zinc, Nitrogen, Nitrigen mineralisation, and pH (3rd row).}
\label{sim3}
\end{center}
\end{figure}

\end{document}